\documentclass{scrartcl}
\KOMAoptions{paper=letter}

\RequirePackage{fix-cm}
\RequirePackage[utf8]{inputenc}
\RequirePackage[T1]{fontenc}
\RequirePackage{microtype} 
\RequirePackage{xspace}
\RequirePackage{amssymb}
\RequirePackage{amsmath}
\RequirePackage{amsthm}
\RequirePackage[hidelinks]{hyperref} 
\RequirePackage[round]{natbib}

\RequirePackage{calc}
\RequirePackage{nicefrac}
\RequirePackage{tikz}
\usetikzlibrary{positioning,arrows,shapes,decorations,calc,fit,arrows.meta,colorbrewer} 
\RequirePackage{colortbl}
\RequirePackage{booktabs}
\RequirePackage{cleveref}
\RequirePackage{array}
\RequirePackage{graphicx}
\RequirePackage{xcolor}
\RequirePackage{enumitem}
\RequirePackage{tabularx}

\AddToHook{cmd/appendix/before}{%
     \crefalias{section}{appendix}%
     \crefalias{subsection}{appendix}
}

\date{} 

\theoremstyle{definition}

\newtheorem{example}{Example}
\theoremstyle{plain}
\newtheorem{theorem}{Theorem}
\newtheorem{lemma}{Lemma}
\newtheorem{proposition}{Proposition}
\newtheorem{corollary}{Corollary}

\newtheoremstyle{bfnote}
{}{}%
{}{}%
{\bfseries}{.}%
{ }%
{\thmname{#1}\thmnumber{ #2}\thmnote{\textnormal{ (#3)}}}
\theoremstyle{bfnote}
\newtheorem{remark}{Remark}

\newcommand{\rsd}{\ensuremath{\mathit{RSD}}\xspace}

\newcommand{\ps}{\ensuremath{\mathit{PS}}\xspace}
\newcommand{\pop}{\ensuremath{\mathit{POP}}\xspace}

\newcommand{\sd}{\ensuremath{\mathit{SD}}\xspace}

\newcommand{\pc}{\ensuremath{\mathit{PC}}\xspace}

\newcommand{\supp}{\mathrm{supp}} 

\newcommand{\Omit}[1]{} 

\usepackage{mathtools, thmtools}

\title{Efficient and Envy-free Random Assignment Beyond Expected Utility}

\author{
Patrick Becker \quad Felix Brandt \quad Satyanand Rammohan\\
Technical University of Munich%
}

\newcommand{\mN}{N} %
\newcommand{\mO}{O} %
\newcommand{\norm}[1]{\left\lVert#1\right\rVert}

\DeclareMathOperator*{\argmin}{arg\,min}

\newcommand{\R}{\mathbb{R}}

\newcommand{\pchz}{\ensuremath{\mathit{PCHZ}}\xspace}

\begin{document}

\maketitle

\begin{abstract}
We consider the random assignment problem with abstract continuous and convex preferences. In particular, we admit preference relations that are not constrained by independence or transitivity. By extending the Hylland--Zeckhauser pseudo-market mechanism, we show that weakly efficient and envy-free random assignments always exist. For preferences that can be represented via skew-symmetric bilinear (SSB) utility functions---which generalize linear expected utility functions---we prove the existence of efficient and approximately envy-free random assignments. Efficient and envy-free random assignments exist under a mild additional assumption on preferences. These findings have notable implications for ordinal random assignment, where ordinal preferences are extended to preferences over lotteries via the pairwise comparison (\pc) extension. While the probabilistic serial rule and popular random assignments frequently and significantly violate \pc-efficiency and \pc-envy-freeness, respectively, random assignments that satisfy both conditions do exist.
\end{abstract}

\section{Introduction} \label{sec:intro}

A central problem in microeconomic theory concerns the fair and efficient assignment of objects to agents based on their preferences over the objects. 
The formal study of the canonical formulation of this problem with $n$ agents and $n$ objects, such that each object must be allocated to exactly one agent, goes back to \citet{Birk46a} and \citet{vNeu53a}. \citet{Kuhn55a} provided an early algorithmic solution to this problem for cardinal preferences. 
When objects are indivisible, it is impossible to deterministically assign objects such that agents with the same preferences receive the same objects. ``Equal treatment of equals'' is usually ensured via randomization, i.e., by assigning lotteries over objects to the agents. An extensive body of research has investigated random assignment rules with respect to properties such as Pareto efficiency, envy-freeness, and strategyproofness.

In contrast to existing work on random assignment, we only impose minimal restrictions on preference relations. In particular, we allow more general preference relations on lotteries than those that can be represented by expected (vNM) utility functions, which \citet{vNM47a} have axiomatically characterized using completeness, transitivity, continuity, and independence.
Independence prescribes that a lottery $x$ is preferred to lottery $y$ if and only if a coin toss between $x$ and a third lottery $z$ is preferred to a coin toss between $y$ and $z$ (with the same coin used in both cases). 
There is experimental evidence that human decision makers systematically violate the independence axiom. Allais's Paradox is the most famous example \citep{Alla53a}.
\citet{Mach83a,Mach89a} and \citet{McCl88a} provide detailed reviews of such violations, including those reported by \citet{KaTv79a}. 
Similarly, a number of scholars have concluded that transitivity can be unnecessarily demanding \citep[see, e.g.,][]{May54a,Fish70c,BaMa88a,Fish91a,Anan93a,Anan09a,HOR19a}. For example, \citet{Anan09a} states that ``once considered a cornerstone of rational choice theory, the status of transitivity has been dramatically reevaluated by economists and philosophers in recent years.''
In a similar vein, \citet[][p.~115]{Fish91a} proclaims that ``transitivity is obviously a great practical convenience and a nice thing to have for mathematical purposes, but long ago this author ceased to understand why it should be a cornerstone of normative decision theory.''

Our most general result merely requires that preferences are continuous and convex. \citet{Sonn71a} has shown that such relations always admit maximal elements, even in the absence of transitivity. Generalizing the pseudo-market mechanism by \citet{HyZe79a}, we prove the existence of weakly efficient and envy-free random assignments. We then move on to the subset of skew-symmetric bilinear (SSB) preferences. SSB preferences, which are significantly more general than vNM preferences, have been characterized by \citet{Fish82c} and \citet{BrBr17a} and admit a convenient compact representation via skew-symmetric matrices \citep[see also][]{Fish84c,Fish88a}. We construct an SSB preference profile that does not admit a (strongly) efficient and envy-free random assignment and suggest two methods to circumvent this negative result. First, we propose a mild additional restriction on preferences that guarantees the existence of efficient and envy-free random assignments. Secondly, for general SSB preferences, we leverage Kakutani's fixed-point theorem to prove the existence of efficient and approximately envy-free random assignments. 

Our findings have noteworthy implications for the well-explored setting where random assignments only depend on the agents' ordinal preferences over objects. 
Identifying ordinal preferences with their canonical SSB utility function leads to the pairwise comparison (\pc) preference extension, which refines the widely studied stochastic dominance (\sd) preference extension. A lottery is \pc-preferred to another lottery if the former is more likely to return a better alternative than the latter. While the probabilistic serial (\ps) rule is \sd-efficient, we point out that it frequently and significantly violates weak \pc-efficiency, i.e., there is another random assignment in which every agent strictly \pc-prefers her lottery to the one she receives under \ps. At the same time, random serial dictatorship and popular random assignments fail to satisfy \pc-envy-freeness. These results lead to the natural question of whether there are rules that satisfy both \pc-efficiency and \pc-envy-freeness, which we answer in the affirmative using our generalization of the Hylland--Zeckhauser pseudo-market.

\subsection*{Related Work}

The work most directly related to ours is the seminal contribution of \citet{HyZe79a}, who introduced the pseudo-market approach to random assignment by combining equal artificial budgets with price-supported lotteries.
Their mechanism can be viewed as an assignment market: competitive prices decentralize individually optimal choices and together with market clearing lead to an efficient allocation \citep{Debr59a}. 

A related line of work, originating with \citet{ArDe54a}, weakens the preference assumptions under which competitive equilibria exist. \citet{MasC74a} and \citet{GaMA75a} establish existence results for economies without transitivity under continuity, convexity, and non-satiation-type assumptions. 
This last assumption is not well suited to the assignment problem, in which agents choose lotteries over a finite set of objects. 
This difficulty has been addressed in different ways: \citet{MasC92a} relaxes the notion of equilibrium,
whereas \citet{Sato10b} provides additional conditions that restore equilibrium existence.
Closest to our approach is the abstract-economy theorem of \citet{ShSo75a}. Our baseline fixed-point argument in \Cref{theorem:market_existence} can be viewed as a specialization of this abstract-economy logic to the Hylland--Zeckhauser pseudo-market.
Our work further contributes to the literature on pseudo-markets as mechanisms for achieving normatively desirable outcomes \citep{HMPY18a, Mira17a}.
The idea to obtain efficient and envy-free allocations via competitive-equilibrium-from-equal-incomes mechanisms has also been applied in other fair division settings \citep{Budi11a,EMZ21a}.

Several papers have explored extensions of classic results for vNM utilities to the more general model of SSB utility functions. For instance, \citet{FiRo86a} have generalized the existence of Nash equilibria \citep{Nash50a}. \citet{ABB14b} have proved an efficiency-welfare theorem connecting efficiency to affine welfare maximization \citep{Carr10a}, which we utilize in the proof of \Cref{theorem:alpha-envy-free}. \citet{BBH15c} have shown that the no-show paradox \citep{Moul88b} disappears for SSB preferences, as a randomized Condorcet extension called \emph{maximal lotteries} satisfies participation. \citet{BrBr17a} characterized a rich subdomain of SSB preferences---the \pc domain---that allows for Arrovian aggregation \citep{Arro51a} and a corresponding social welfare function. \citet{ABBB15a} and \citet{BLS22c} have shown negative results for efficient and strategyproof social choice functions, similar in spirit to classic results by \citet{Gibb77a} and \citet{Hyll80a}.

\section{Preliminaries}

Let $\mN = \{1,\dots,n\}$ be a set of $n$ agents and $\mO$ be a set of $n$ objects.
A \emph{deterministic assignment (or pure matching)} is a %
permutation matrix in $\R^{n\times n}$.
A \emph{random assignment} is a probability distribution over deterministic assignments, which we represent as a bistochastic matrix $X = (x_i(o))_{i \in N, o \in O}$ where $x_i(o)$ is the probability with which agent $i$ receives object $o$.\footnote{A matrix is bistochastic if all entries are non-negative and every row and every column sums up to $1$.}
The set of all random assignments is denoted by $\mathcal{M}$.
By the Birkhoff--von Neumann decomposition, every bistochastic matrix can be written as a probability distribution over deterministic assignments \citep{Birk46a,vNeu53a}.
The set of all probability distributions (or lotteries) over $\mO$ is denoted by $\Delta$. For a random assignment $X\in\mathcal{M}$ and $i\in N$, we write $x_i$ for the $i$th row of $X$, i.e., the lottery over $\mO$ assigned to agent $i$; thus, $x_i\in\Delta$.
A lottery is \emph{degenerate} if it puts all probability on a single object.

As an example, consider the following random assignment $X$ where $n=3$ and $\mO=\{a,b,c\}$.
	\[X = \begin{pmatrix}
	\nicefrac{1}{3} & \nicefrac{2}{3} & 0\\
	\nicefrac{2}{3} & \nicefrac{1}{12} & \nicefrac{1}{4} \\
	0 & \nicefrac{1}{4} & \nicefrac{3}{4}
	\end{pmatrix}
	\]
Here, $x_1(a)=\nicefrac{1}{3}$ and $x_3=(0,\nicefrac{1}{4}, \nicefrac{3}{4})$.

Every agent $i\in \mN$ has an asymmetric, binary preference relation $\succ_i$ over the elements of $\Delta$. As a consequence, there are no externalities, and agents are only concerned with their own assignment.

For lotteries $x,y\in\Delta$, write $x\sim y$ if neither $x\succ y$ nor $y\succ x$, and $x\succsim y$ if either $x\succ y$ or $x\sim y$. %
A \emph{preference profile} $(\succ_1, \dots, \succ_n)$ is an $n$-tuple of preference relations.

Throughout, we assume that agents have continuous and convex preferences in the sense that the graph of $\succ$ is open and both weak and strict upper contour sets are convex.
More precisely, we demand that
\[
G(\succ) \coloneqq \{(x, y) \in \Delta \times \Delta \colon x \succ y\} \text{ is open.}
\tag{continuity}
\]
Furthermore, we require that for all $x\in\Delta$,
\[
\text{$W(x) \coloneqq \{y\in\Delta \colon y\succsim x\}$ and $U(x) \coloneqq \{y\in\Delta \colon y\succ x\}$ are convex.}\tag{convexity}
\]
Continuity implies that both strict upper contour sets and strict lower contour sets are open.

\citet{Sonn71a} has shown that convexity of strict upper contour sets and openness of strict lower contour sets suffice to guarantee the existence of maximal elements in every non-empty, compact, and convex set of lotteries, even when preferences are intransitive \citep[see also][]{Berg92a,Llin98a}. 
Furthermore, demanding that weak upper contour sets are convex ensures that sets of maximal elements are convex.
Continuity ensures that strict preference comparisons are robust to small perturbations of both lotteries.

A subset of the domain of continuous and convex preferences admits a representation by skew-symmetric bilinear (SSB) utility functions.
A preference relation can be expressed by an SSB utility function $\phi\colon\Delta\times\Delta\rightarrow \mathbb R$ if for all $x,y\in\Delta$,
\[
	x\succ y \text{ if and only if }\phi(x,y) > 0\text.
	\tag{SSB preferences}
\]
Skew-symmetry requires that $\phi(x,y) = - \phi(y,x)$ for all $x,y\in\Delta$ and bilinearity that $\phi$ is linear in both arguments.
Note that, by skew-symmetry, linearity in the first argument implies linearity in the second argument and that, due to bilinearity, $\phi$ is completely determined by its function values for degenerate lotteries.
Thus, with slight abuse of notation, we will represent every SSB utility function $\phi$ by a skew-symmetric matrix $\phi\in \mathbb{R}^{O\times O}$.
As mentioned in \Cref{sec:intro}, SSB utility was introduced by \citet{Fish82c} as a generalization of classic linear expected utility that does not require the somewhat controversial axioms of independence and transitivity \citep[see also][]{Fish84c,Fish88a}. When $\phi$ is separable, i.e., $\phi(x,y) = u(x) - u(y)$ for some $u\in \mathbb{R}^O$, the preferences represented by $\phi$ boil down to vNM expected utility with utility function $u$.
Through the representation of ${\succ}$ as a skew-symmetric matrix $\phi$, it becomes apparent that the Minimax Theorem \citep{vNeu28a} implies the existence of maximal elements of $\succ$. This was noted by \citet[Theorem 4]{Fish84c} and already follows from \citet{Sonn71a}.

The central question pursued in this paper is under which conditions on preferences the existence of efficient and envy-free random assignments can be guaranteed.

A random assignment $X\in\mathcal{M}$ is efficient if there is no $Y\in\mathcal{M}$ such that
\[
	y_i \succsim_i x_i \text{ for all }i\in N \text{ and }y_i \succ_i x_i \text{ for some } 
	i\in N\text.
	\tag{efficiency}
\]
If there is such an $Y\in\mathcal{M}$, we say that $Y$ \emph{Pareto dominates} $X$.

$X\in\mathcal{M}$ is weakly efficient if there is no $Y\in\mathcal{M}$ such that
\[
	y_i \succ_i x_i \text{ for all }i\in N\text.
	\tag{weak efficiency}
\]
If there is such an $Y\in\mathcal{M}$, we say that $Y$ \emph{strongly Pareto dominates} $X$.

A random assignment $X\in\mathcal{M}$ is envy-free if 
\[
	x_i\succsim_i x_j \text{ for all }i,j\in N\text.
	\tag{envy-freeness}
\]

\section{Incompatibility of Efficiency and Envy-Freeness} \label{section:counterexample}

We first show that efficiency and envy-freeness cannot always be satisfied simultaneously, even when preferences are represented by SSB utility functions. This is demonstrated by the following preference profile for three agents.

\begin{restatable}{theorem}{CounterExample} \label{theorem:counter-example}
	Let $O=\{a,b,c\}$ and consider the preference profile $(\succ_1, \succ_2, \succ_3)$, represented by the SSB utility functions
	\begin{equation*}
        \phi_1 = \begin{pmatrix}
            0 & 1 & 1 \\
            -1 & 0 & 1 \\
            -1 & -1 & 0
        \end{pmatrix}\text, \quad
        \phi_2 = \begin{pmatrix}
            0 & 1 & 1 \\
            -1 & 0 & 1 \\
            -1 & -1 & 0
        \end{pmatrix}\text{, and} \quad
        \phi_3 = \begin{pmatrix}
            0 & -1 & 1 \\
            1 & 0 & 0 \\
            -1 & 0 & 0
        \end{pmatrix}.
    \end{equation*}
	For this profile, no random assignment satisfies both efficiency and envy-freeness.
\end{restatable}

The proof of \Cref{theorem:counter-example} can be found in \Cref{appendix:counter-example}.
The intransitivity, even over degenerate lotteries, of agent~$3$'s preference relation plays a key role in the proof. 
In particular, when restricted to degenerate lotteries (identified with pure objects in $O$), agents~$1$ and~$2$ share the transitive strict order $a\succ b\succ c$, while agent~$3$'s preference is given by $b\succ_3 a$, $a\succ_3 c$, and $b\sim_3 c$.
This has the consequence that agent~$3$ is indifferent between the lottery $x_3 = (0, \nicefrac12, \nicefrac12)$ and every other lottery.

\section{A Generalized Pseudo-Market} \label{section:pseudo-markets}

Pseudo-markets provide a natural framework for decentralizing allocation problems without monetary transfers.
Agents are endowed with artificial budgets and use these to purchase probabilistic shares of indivisible objects; in the assignment problem, each object is available in unit supply. 
\citet{HyZe79a} proposed a competitive mechanism for vNM utilities and showed that every random assignment at market equilibrium satisfies efficiency and envy-freeness, assuming agents are given equal budgets and each agent selects a cost-minimal element from the set of budget-feasible utility maximizers. 
We revisit this construction at the level of abstract preferences over lotteries, and prove that our versions of continuity and convexity suffice to retain existence of a pseudo-market equilibrium.
\Cref{theorem:counter-example} already suggests that the axiomatic properties of the equilibrium cannot carry over verbatim to our more general preference model; nevertheless, we prove that the equilibrium assignment satisfies \emph{weak} efficiency in conjunction with envy-freeness.

In \Cref{section:SSB_with_strict_separation}, we further identify a preference domain restriction which restores (strong) efficiency of the pseudo-market equilibrium by ruling out obstructions of the type seen in \Cref{theorem:counter-example}.
As we show, this domain restriction still includes all weak preference relations that are transitive on degenerate lotteries.
Hence, this domain remains a significant generalization of the vNM domain originally considered by \citet{HyZe79a}.

\medskip

Assume each agent $i \in N$ receives a budget $b_i>0$ to buy probability shares of the objects in $O$; assume also that $\sum_{i\in N} b_i=1$.
There is a common price vector $p \in \R^n_{\geq 0}$, where the component $p_o$ denotes the price of one full probability share of object $o \in O$.
The price vectors are normalized so that they lie in a compact \emph{price set} $P \subseteq \R^n_{\geq0}$ that satisfies $\min_{o \in O} p_o = 0$ for all $p \in P$.\footnote{Since this price set is not convex, we will obtain prices from a compact and convex parameter space via a continuous normalization map in a manner similar to \citet{HyZe79a}.}
Given a price vector $p\in P$ and a lottery $x \in \Delta$, the \emph{cost} of $x$ under the price vector $p$ is given by the inner product $p \cdot x$.
Each agent $i$ can choose a lottery whose cost does not exceed her budget $b_i$. 
Therefore, the \emph{budget set} of agent $i$ is defined as
\begin{equation*}
	B_i(p) \coloneqq \{x\in \Delta\colon p\cdot x \le b_i\}.
	\tag{budget set}
\end{equation*}
We assume that each agent $i$ is rational in the sense that she selects, among the lotteries in $B_i(p)$, a maximal element according to her preference relation $\succ_i$. 
The \emph{demand set} of agent $i$ at prices $p$ is the set of maximal affordable lotteries,
\begin{equation*}
	D_i(p) \coloneqq \{x\in B_i(p)\colon x \succsim_i y \text{ for all } y\in B_i(p)\}.
	\tag{demand set}
\end{equation*}

A pseudo-market equilibrium is then a pair consisting of a price vector $p^\ast$ and a profile $X^\ast = (x^\ast_i)_{i \in N}$ of lotteries with $x^\ast_i \in D_i(p^\ast)$ for each agent $i$, such that aggregate demand exactly exhausts the unit supply of every object.

Formally, a pair $(X^\ast,p^\ast) \in \Delta^n \times P$ is a \emph{pseudo-market equilibrium} if it satisfies
\begin{enumerate}[leftmargin=*,label=\textit{(\roman*)}]
	\item $x_i^\ast \in D_i(p^\ast)$ for every agent $i \in N$
	\hfill \textnormal{(rationality)}	
	\item $\sum_{i\in N} x_i^\ast(o) = 1$ for every object $o \in \mO$
	\hfill \textnormal{(market-clearing)}
\end{enumerate}
By the market-clearing condition, we can identify $X^\ast = (x^\ast_i)_{i \in N}$ as a random assignment, that is,  $X^\ast \in\mathcal{M}$.

\begin{restatable}{theorem}{MarketExistence}\label{theorem:market_existence}
	For every continuous and convex preference profile, a pseudo-market equilibrium exists.
\end{restatable}

The proof of \Cref{theorem:market_existence}, using Kakutani's fixed-point theorem, is deferred to \Cref{appendix:pseudomarket}.
The key ingredients are: \emph{(i)} the budget set correspondence $B_i(\cdot)$ is continuous with non-empty, compact and convex values, \emph{(ii)} the demand correspondence $D_i(\cdot)$ is upper hemicontinuous with non-empty, compact and convex values, and \emph{(iii)} the normalized price-selection correspondence is upper hemicontinuous with non-empty, compact and convex values.
Kakutani's theorem then yields the existence of a pseudo-market equilibrium. 

Even in our generalized preference model, the pseudo-market equilibrium satisfies envy-freeness, under the assumption that agents' budgets are equal.

\begin{proposition} \label{proposition:equilibrium-envy-freeness}
	Suppose agents have equal budgets, i.e., $b_i=b_j$ for all $i,j\in N$. Then for every pseudo-market equilibrium $(X^\ast, p^\ast)$, the random assignment $X^\ast$ is envy-free.
\end{proposition}
\begin{proof}
	Let $(X^\ast,p^\ast)$ be a pseudo-market equilibrium and $i, j \in N$ any two agents. 
	By assumption we have $b_i = b_j$.
	Since $x_j^\ast \in B_j(p^\ast)$, we have $p^\ast \cdot x_j^\ast \le b_j = b_i$.
	This implies $x_j^\ast \in B_i(p^\ast)$, i.e., $x^\ast_j$ is also affordable for agent $i$. 
	Since $x_i^\ast \in D_i(p^\ast)$, $x^\ast_i$ is maximal in $B_i(p^\ast)$, and we have $x_i^\ast \succsim_i x_j^\ast$.
	Hence, $X^\ast$ is envy-free.
\end{proof}

As mentioned, the pseudo-market equilibrium does not directly inherit efficiency from the Hylland--Zeckhauser mechanism for vNM utilities.
The key property that holds for vNM utilities and fails in our more general model is: if $x_i, y_i \in B_i(p)$, with $x_i$ maximal and $y_i$ not maximal, then $x_i \succ_i y_i$ (this follows directly from the linearity of vNM utility functions). 
Even if we strengthen the rationality assumption to require that agents break ties between multiple maximal affordable lotteries by minimizing the cost $p \cdot x$ (following \citealp{HyZe79a}), the equilibrium assignment may be Pareto-dominated.
This violation can be seen as a consequence of \Cref{theorem:counter-example}.

\begin{example}
	Consider the preference profile for $n = 3$ from \Cref{theorem:counter-example}, and let all agents have equal budgets ($b_i =\nicefrac{1}{3}$ for all $i \in N$).
	A pseudo-market equilibrium is given by the pair $(X^\ast, p^\ast)$ with 
	\[
		X^\ast=
		\begin{pmatrix}
		\nicefrac{1}{2} & \nicefrac{1}{4} & \nicefrac{1}{4}\\
		\nicefrac{1}{2} & \nicefrac{1}{4} & \nicefrac{1}{4}\\
		0 & \nicefrac{1}{2} & \nicefrac{1}{2}
		\end{pmatrix}
		\quad\text{and}\quad
		p^\ast=\left(\nicefrac{5}{9},\nicefrac{2}{9},0\right).
	\]
	This pseudo-market equilibrium even satisfies the stronger rationality requirement of cost-minimization for each agent.
	Nevertheless, $X^\ast$ is Pareto dominated by the random assignment
	\[	
		Y=
		\begin{pmatrix}
		\nicefrac{1}{2} & \nicefrac{1}{2} & 0\\
		\nicefrac{1}{2} & \nicefrac{1}{2} & 0\\
		0 & 0 & 1
		\end{pmatrix}.
	\]
\end{example}

On the other hand, as we show below, the pseudo-market equilibrium satisfies weak efficiency.

\begin{proposition} \label{proposition:equilibrium-wefficiency}
	For every pseudo-market equilibrium $(X^\ast, p^\ast)$, the random assignment $X^\ast$ is weakly efficient.
\end{proposition}
\begin{proof}
	Let $(X^\ast,p^\ast)$ be a pseudo-market equilibrium, and suppose, towards a contradiction, that $X^\ast$ is not weakly efficient. 
	Then there exists another random assignment $Y$ such that
	\begin{equation*}
		y_i\succ_i x_i^\ast \qquad \text{for all }i \in N.
	\end{equation*}
	Since $x_i^\ast\in D_i(p^\ast)$, we have $p^\ast \cdot x^\ast_i \leq b_i$, and $x^\ast_i$ is maximal in the budget set $B_i(p^\ast)$. 
	Hence, $y_i \succ_i x_i^\ast$ implies $y_i \notin B_i(p^\ast)$. 
	In particular, $p^\ast\cdot y_i > b_i$ holds for every agent $i \in N$.
	Summing over agents yields
	\begin{equation*}
		p^\ast\cdot \sum_{i\in N}y_i > \sum_{i \in N} b_i \geq p^\ast \cdot \sum_{i\in N} x^\ast_i.
	\end{equation*}
	However, since $Y$ is a random assignment and $X^\ast$ satisfies market-clearing, we have $\sum_{i\in N}y_i = \mathbf{1} = \sum_{i\in N} x^\ast_i$, and thus both sides of the inequality coincide, a contradiction.
	Thus, no such $Y$ exists, and $X^\ast$ is weakly efficient.
\end{proof}

As a consequence of \Cref{proposition:equilibrium-envy-freeness,proposition:equilibrium-wefficiency}, the existence of pseudo-market equilibria, as shown by \Cref{theorem:market_existence}, immediately yields random assignments that are both weakly efficient and envy-free.
\begin{corollary} \label{corollary:weak-eff-plus-EF}
	For every continuous and convex preference profile, there exists a random assignment satisfying weak efficiency and envy-freeness.
\end{corollary}

\subsection{SSB Preferences with Strict Maximality}\label{section:SSB_with_strict_separation}

\Cref{theorem:counter-example} has established that even the restricted domain of SSB preferences allows large indifference regions that prevent pseudo-market equilibria from being efficient. 
We now identify a further domain restriction within SSB preferences that precisely rules out this obstruction. 
Specifically, we impose a \emph{strict maximality} condition that requires, for each agent $i \in N$, that the strict part $\succ_i$ separates the maximal elements in $\Delta$ from the rest of the simplex $\Delta$. 
Denote the set of \emph{maximal elements} of the preference relation $\succ_i$ in the simplex $\Delta$ by
\[
	M_i \coloneqq \{x\in\Delta\colon x \succsim_i y \text{ for all }y\in\Delta\}
	\tag{maximal elements}
\]
Then $\succ_i$ satisfies \emph{strict maximality} if 
\[
	x \succ_i y \text{ for all }x\in M_i,\ y\notin M_i\text.
	\tag{strict maximality}
\]

Within the domain of SSB preferences, the strict maximality condition can be characterized by the following equivalent condition on the matrix representation $\phi_i$.

\begin{restatable}{proposition}{strictMaximalityCharacterization} \label{proposition:characterization_SSB_strict_maximality}
	Let $\phi_i$ be an SSB matrix. 
	The preference relation $\succ_i$ represented by $\phi_i$ satisfies strict maximality if and only if there exists a reordering of the objects $\mO$ after which 
	\[
		\phi_i =
		\begin{pmatrix}
		0 & Q\\
		-Q^\top & C
		\end{pmatrix},
	\]
	where every entry of $Q$ is strictly positive.
\end{restatable}

A proof of \Cref{proposition:characterization_SSB_strict_maximality} can be found in \Cref{appendix:strictness_assumption}.
Note that by this characterization, strict maximality implies that the set $M_i$ of maximal elements is a face of the simplex, with vertices given exactly by the set $T \subseteq O$ of objects corresponding to the top-left zero block of the reordered matrix $\phi_i$.

\begin{corollary}
	Let~$\succ$ be a preference relation represented by an SSB utility function. If~$\succsim$ is transitive on degenerate lotteries, then~$\succ$ satisfies strict maximality.
\end{corollary}

It turns out that the pseudo-market equilibrium returns a random assignment satisfying efficiency and envy-freeness in the domain of SSB preferences satisfying strict maximality.
In addition to the rationality and market-clearing conditions of our defined pseudo-market equilibrium, we require that each agent $i$ selects a cost-minimizing element from her demand set $D_i(p)$.
That is to say, we refine the demand correspondence $D_i(\cdot)$. 
Given a price vector $p$, we assume that agent $i$ selects from
\[
	\widehat D_i(p) \coloneqq \argmin_{x\in D_i(p)} p\cdot x.
\]

Under SSB preferences, the correspondence $\widehat{D}_i(\cdot)$ inherits the properties required for our fixed-point existence argument, which is otherwise similar to the proof of \Cref{theorem:market_existence}.

\begin{restatable}{theorem}{CostMinimizationExistence} \label{theorem:existence-with-cost-minimization}
	For every SSB preference profile, a pseudo-market equilibrium with cost-minimization exists.
\end{restatable}

The proof of \Cref{theorem:existence-with-cost-minimization} is given in \Cref{appendix:strictness_assumption}.

\begin{remark}
	The assumption of SSB preferences cannot be dropped entirely; one can construct non-SSB preference profiles that, in spite of \Cref{theorem:market_existence}, do not admit pseudo-market equilibria with the additional cost-minimization requirement.
	Such a profile is constructed in \Cref{example:non-existence-cost-minimization} in \Cref{appendix:strictness_assumption}.
\end{remark}

For SSB preferences satisfying strict maximality, we strengthen the efficiency guarantee of \Cref{proposition:equilibrium-wefficiency} from weak efficiency to (strong) efficiency for the refined version of the pseudo-market equilibrium.
The necessity of tie-breaking by cost-minimization to obtain efficiency of the equilibrium assignment is already recognized by \citet{HyZe79a}. 

\begin{restatable}{proposition}{PseudoMarketStrongEfficiency} \label{proposition:equilibrium-sefficiency}
	For every SSB preference profile satisfying strict maximality, and every pseudo-market equilibrium $(X^\ast, p^\ast)$ with cost-minimization, the random assignment $X^\ast$ is efficient.
\end{restatable}
\begin{proof}
	Let $(X^\ast, p^\ast)$ be a pseudo-market equilibrium with cost-minimization.
	For each agent $i \in N$ and each lottery $y \in \Delta$, $y \succ_i x^\ast_i$ implies $y \notin B_i(p^\ast)$; otherwise, the maximality of $x^\ast_i$ in $B_i(p^\ast)$ is violated.
	Hence, 
	\[
		y \succ_i x^\ast_i \quad\Longrightarrow\quad p^\ast\cdot y > p^\ast\cdot x_i^\ast.
	\]
	We also claim that
	\[
		y \succsim_i x^\ast_i \quad\Longrightarrow\quad p^\ast\cdot y\ge p^\ast\cdot x_i^\ast.
	\]
	Suppose, towards a contradiction, that $y \succsim_i x^\ast_i$ and $p^\ast\cdot y<p^\ast\cdot x_i^\ast$.
	Since $x_i^\ast \in B_i(p^\ast)$, this implies $y \in B_i(p^\ast)$ as well.
	We distinguish two cases.

	First, suppose that $x_i^\ast\in M_i$, i.e., the lottery $x^\ast_i$ is maximal for agent $i$ in $\Delta$. 
	Since we assumed $y \succsim_i x^\ast_i$, we get $y \sim_i x^\ast_i$, which by strict maximality implies $y \in M_i$, i.e., $y$ is also maximal in $\Delta$.
	Hence, $y$ is also maximal in the affordable set $B_i(p^\ast) \subseteq \Delta$, i.e., $y\in D_i(p^\ast)$.
	But then $p^\ast\cdot y<p^\ast\cdot x_i^\ast$ contradicts that $x_i^\ast$ is a cost-minimal element of $D_i(p^\ast)$.

	Second, suppose that $x_i^\ast\notin M_i$. 
	Choose some maximal element $m\in M_i$. 
	By strict maximality, $m \succ_i x^\ast_i$.
	For $\varepsilon\in(0,1)$, define $z^\varepsilon \coloneqq (1-\varepsilon)y +\varepsilon m$.
	By bilinearity,
	\[
		\phi_i(z^\varepsilon,x_i^\ast) = (1-\varepsilon)\phi_i(y,x_i^\ast) + \varepsilon\phi_i(m,x_i^\ast) > 0,
	\]
	where we used $\phi_i(y,x_i^\ast) \geq 0$ and $\phi_i(m,x_i^\ast) > 0$.
	Thus, $z^\varepsilon \succ_i x_i^\ast$.
	Moreover, since $p^\ast\cdot y < p^\ast\cdot x_i^\ast$, the continuity of the dot product implies that we can choose $\varepsilon$ sufficiently small so that $p^\ast\cdot z^\varepsilon<p^\ast\cdot x_i^\ast\le b_i$.
	Hence, $z^\varepsilon\in B_i(p^\ast)$, contradicting the maximality of $x_i^\ast$ in $B_i(p^\ast)$.
	This proves the claim.

	Finally, suppose, towards a contradiction, that $X^\ast$ is not efficient.
	Then there exists a random assignment $Y$ such that $y_i \succsim_i x^\ast_i$ for all $i \in N$, and $y_i\succ_i x^\ast_i$ for some $i \in N$.
	By the claim, 
	\begin{align*}
		p^\ast \cdot y_i &\geq p^\ast \cdot x^\ast_i \qquad \text{for all }i\in N, \\
		p^\ast \cdot y_i &> p^\ast \cdot x^\ast_i \qquad \text{for some }i\in N.
	\end{align*}
	Summing over agents gives
	\[
		p^\ast\cdot\sum_{i\in N}y_i > p^\ast\cdot\sum_{i\in N}x_i^\ast.
	\]
	However, since $Y$ is a random assignment and $X^\ast$ satisfies market-clearing, we have $\sum_{i\in N}y_i = \mathbf 1 = \sum_{i\in N}x_i^\ast$, and thus both sides of the inequality coincide, a contradiction.
	Thus, no such $Y$ exists, and $X^\ast$ is efficient.
\end{proof}

Combining the existence of a pseudo-market equilibrium with cost minimization (\Cref{theorem:existence-with-cost-minimization}), envy-freeness under equal budgets (\Cref{proposition:equilibrium-envy-freeness}), and the strengthened efficiency guarantee of \Cref{proposition:equilibrium-sefficiency}, we obtain the existence of efficient and envy-free random assignments for the domain of SSB preferences satisfying strict maximality.

\begin{corollary}
	For every SSB preference profile satisfying strict maximality, there exists a random assignment satisfying efficiency and envy-freeness.
\end{corollary}

\begin{remark}
	Neither \Cref{theorem:existence-with-cost-minimization} nor \Cref{proposition:equilibrium-sefficiency} rely \emph{per se} on preferences being represented by SSB utility functions. 
	Beyond strict maximality and the maintained continuity and convexity assumptions, our argument only requires the following axiom of \citet[Axiom D2]{Fish84c}: for all $x,y,z\in \Delta$ with $x\succ y$ and $y \succsim z$, it holds that $(\lambda x + (1-\lambda) y) \succ z$ for every $\lambda \in (0, 1]$.
	SSB preferences satisfy this axiom by bilinearity, but it can be satisfied by non-SSB preference relations as well.
	One natural class of non-SSB preferences that satisfies both strict maximality and Axiom D2 are those represented by \emph{$\ell_p$-disutilities} for $p \geq 1$.
	Such preferences allow for strictly maximal elements in the interior of the simplex $\Delta$, which is ruled out by SSB preferences (see \Cref{proposition:characterization_SSB_strict_maximality}).
\end{remark}

\section{Efficiency and $\alpha$-Envy-Freeness for SSB Preferences}
In \Cref{section:counterexample}, we showed that the existence of efficient and envy-free random assignments is unattainable in the full domain of SSB preferences.
We bypassed this impossibility in \Cref{section:SSB_with_strict_separation} by restricting the SSB domain using strict maximality.
In this section, we show that \Cref{theorem:counter-example} can be circumvented in the full SSB domain by relaxing envy-freeness.
In particular, we prove, for the full SSB domain, the existence of random assignments that are efficient and induce arbitrarily small envy among agents.

Let $\phi = (\phi_1, \dots, \phi_n)$ be a profile of SSB utility functions, and let $\alpha$ be a positive constant.
A random assignment $X \in \mathcal{M}$ is $\alpha$-envy-free if 
\[
\phi_i(x_j, x_i) \leq \alpha \text{ for all }i, j \in N\text.
\tag{$\alpha$-envy-freeness}
\]

We provide a fixed-point argument inspired by \citet{CoTa21a}, while leveraging an efficiency-welfare theorem for SSB preferences due to \citet{ABB14b}.
This theorem is restated below.
\begin{theorem}[\citealp{ABB14b}] \label{theorem:efficiency-welfare}
	Let $(\phi_1, \dots, \phi_n)$ be a profile of SSB utility functions. 
	A random assignment $X \in \mathcal{M}$ is efficient if and only if there exists a vector $\omega \in \R^n_{>0}$ of positive weights that satisfies
	\[
		\sum_{i \in N} \omega_i \phi_i(x_i, y_i) \geq 0 \qquad \text{for all } Y \in \mathcal{M}.
	\]
\end{theorem}

The main theorem of this section is the following.
\begin{restatable}{theorem}{AlphaEnvyFree} \label{theorem:alpha-envy-free}
	Let $(\phi_1, \dots, \phi_n)$ be a profile of SSB utility functions. For any $\alpha > 0$, there exists a random assignment that is efficient and $\alpha$-envy-free.
\end{restatable}

Let $\alpha > 0$ be arbitrary.
The high-level proof idea is to construct a Kakutani mapping $\Gamma$ which acts on a pair $(X, \omega)$ of a random assignment $X$ and a weight vector $\omega \in \R^n_{>0}$. At a fixed point, we show that the random assignment $X$ is efficient and $\alpha$-envy-free. 

Fix $\varepsilon>0$\footnote{We later specify $\varepsilon$ such that $\Omega_\varepsilon$ is non-empty.} and define
\begin{equation*}
	\Omega_\varepsilon \coloneqq \left\{ \omega \in \mathbb{R}^n \colon \omega_i \geq \varepsilon \text{ for all } i \text{ and } \sum_{i \in N} \omega_i = 1 \right\}.
\end{equation*}
The set $\Omega_\varepsilon$ is compact and convex. Each vector $\omega\in \Omega_\varepsilon$ assigns every agent a strictly positive weight, uniformly bounded away from zero. The lower bound $\varepsilon$ is introduced for technical reasons:
the set of all strictly positive normalized weights is relatively open in the unit simplex and therefore not compact. Since Kakutani's fixed-point theorem requires a non-empty compact convex domain, we work on $\Omega_\varepsilon$ instead. The value of $\varepsilon$ will later be chosen sufficiently small as a function of the desired approximation parameter~$\alpha$.

For a weight vector $\omega \in \Omega_\varepsilon$, define 
\begin{equation*}
	F(\omega) \coloneqq \left\{ X \in \mathcal{M} \colon \sum_{i \in N} \omega_i \phi_i(x_i, y_i) \geq 0 \text{ for all } Y \in \mathcal{M} \right\}.
\end{equation*}
Since every $\omega \in \Omega_\varepsilon$ is strictly positive, \Cref{theorem:efficiency-welfare} implies that every $X \in F(\omega)$ is efficient.
So, for a given $\omega$, the correspondence $F$ maps to a set of efficient random assignments.

On the other hand, for a given random assignment $X$, we will update the welfare weights $\omega$ so that weights of envious agents in $X$ are increased at the expense of non-envious agents.
This weight update function can be described in two steps.
First, we increase the weight $\omega_i$ associated with each agent $i$ in proportion to the maximum envy (above the $\alpha$ threshold) she suffers against any other agent $j$ in the assignment $X$:
\begin{align*}
	\nu_i(X) &\coloneqq \max_{j \in N} \max\{\phi_i(x_j, x_i) -\alpha, 0\}, \\
	\mu(X, \omega) &\coloneqq \omega + \nu(X).
\end{align*}
Note that the updated weight vector $\mu' = \mu(X, \omega)$ may leave $\Omega_\varepsilon$ by violating the constraint that weights sum to 1.
Therefore, the second step is to project $\mu'$ back into $\Omega_\varepsilon$ via the Euclidean projection:
\begin{align*}
	g(X, \omega) 
	&\coloneqq \operatorname{proj}_{\Omega_\varepsilon} \mu(X, \omega) \\
	&= \argmin_{\omega' \in \Omega_\varepsilon} \norm{\omega' - \mu(X, \omega)}.
\end{align*}

Finally, we combine the two mappings $F$ and $g$ in the correspondence $\Gamma: \mathcal{M} \times \Omega_\varepsilon \rightrightarrows \mathcal{M} \times \Omega_\varepsilon$ with
\begin{equation*}
	\Gamma(X, \omega) \coloneqq F(\omega) \times \{ g(X, \omega) \}.
\end{equation*}

\begin{restatable}{proposition}{FixedPointExistence} \label{proposition:fixed-point-existence}
	The correspondence $\Gamma$ has a fixed point, i.e., there exists $(X^\ast, \omega^\ast) \in \mathcal{M} \times \Omega_\varepsilon$ such that $(X^\ast, \omega^\ast) \in \Gamma(X^\ast, \omega^\ast)$. 
\end{restatable}

The proof of \Cref{proposition:fixed-point-existence} uses Kakutani's fixed-point theorem, and is given in \Cref{appendix:alpha-envy-free}.
We further claim that the random assignment $X^\ast$ at the guaranteed fixed point satisfies both efficiency and $\alpha$-envy-freeness (\Cref{proposition:fixed-point-efficiency,proposition:fixed-point-envy-freeness}), which, together with \Cref{proposition:fixed-point-existence}, imply \Cref{theorem:alpha-envy-free}.

\begin{proposition} \label{proposition:fixed-point-efficiency}
	If $(X^\ast, \omega^\ast)$ is a fixed point of $\Gamma$, then the random assignment $X^\ast$ satisfies efficiency.
\end{proposition}

This follows directly from \Cref{theorem:efficiency-welfare}, since $X^\ast \in F(\omega^\ast)$.

\begin{restatable}{proposition}{FixedPointEnvyFreeness} \label{proposition:fixed-point-envy-freeness}
	If $(X^\ast, \omega^\ast)$ is a fixed point of $\Gamma$, then the random assignment $X^\ast$ satisfies $\alpha$-envy-freeness.
\end{restatable}

The proof of \Cref{proposition:fixed-point-envy-freeness} is provided in \Cref{appendix:alpha-envy-free}.

One might hope that the existence of random assignments that satisfy arbitrarily good approximations of envy-freeness while maintaining efficiency would imply the existence of an exactly envy-free assignment that is also efficient.
Indeed, the fact that the space of random assignments, namely the Birkhoff polytope, is compact, implies that any sequence of efficient and $\alpha$-envy-free assignments with $\alpha \to 0$ has a convergent subsequence whose limit assignment is exactly envy-free.
Unfortunately, it is not guaranteed that the limit assignment remains efficient; nevertheless, the limit assignment is guaranteed to be weakly efficient.
These observations are due to the following result, whose proof is deferred to \Cref{appendix:alpha-envy-free}.

\begin{restatable}{proposition}{SSBEfficientSetOpen} \label{proposition:eff-not-closed}
    Given an SSB preference profile,
	\begin{itemize}
		\item[(i)] the set of efficient random assignments need not be closed.
		\item[(ii)] the set of weakly efficient assignments is closed.
	\end{itemize}
\end{restatable}

This fact implies the existence of random assignments satisfying weak efficiency and exact envy-freeness, by taking the limit of a convergent sequence of efficient and $\alpha$-envy-free random assignments, with $\alpha \to 0$.
We remark that this implication coincides with \Cref{corollary:weak-eff-plus-EF} in the SSB domain.

\begin{corollary}\label{corollary:welfare_weight_limit}
    For any SSB preference profile, there exists a random assignment satisfying weak efficiency and envy-freeness.
\end{corollary}

\section{Ordinal Random Assignment} \label{section:ordinal-ra}

The results from the previous sections have interesting consequences for the widely studied problem of \emph{ordinal} random assignment. 
An ordinal random assignment rule returns a random assignment for each profile of complete and transitive preference relations over objects $O$ (rather than lotteries over objects $\Delta$). Random assignment rules are typically evaluated in terms of efficiency, envy-freeness, and strategyproofness by systematically extending the preferences over objects to preferences over lotteries. The most widely studied way of extending preferences to lotteries is based on \emph{stochastic dominance (\sd)}, where one lottery \emph{stochastically dominates} another lottery if the former yields at least as much expected utility as the latter for every vNM function that is consistent with the ordinal preferences \citep[see, e.g.,][]{PoSc86a,BoMo01a}. The \sd relation is generally incomplete, i.e., there are incomparable pairs of lotteries. This leads to two notions of envy-freeness---strong and weak \sd-envy-freeness---depending on how incomparabilities are treated.

\citet{ABB14b} initiated the study of extending preferences via the \emph{pairwise comparison (\pc)} relation, a complete refinement of the \sd relation admitting a natural interpretation. For an asymmetric preference relation $\succ$ over $O$ and two lotteries $x, y \in \Delta$, $x$ is \emph{\pc-preferred} to $y$, written $x \succsim^\pc y$, if and only if
\[
\sum_{o, \hat o\in O\colon o \succ \hat o} x(o) \cdot y(\hat o) \geq \sum_{o, \hat o\in O\colon o \succ \hat o} y(o) \cdot x(\hat o). \tag{\pc}
\]
Lottery $x$ is preferred to lottery $y$ if it is as least as likely that $x$ yields a better object than $y$ than \emph{vice versa}. Alternatively, the terms in the inequality above can be associated with \emph{ex ante} regret (the probability of \emph{ex post} regret): a decision maker would less frequently regret choosing $x$ than $y$ \emph{ex post}.
\pc preferences have been considered in decision theory \citep{Blyt72a,Pack82a,Blav06a}. \citet{Pack82a} calls them the \emph{rule of expected dominance} and \citet{Blav06a} refers to them as a \emph{preference for the most probable winner}. 
\citet{ABB14b,ABBB15a}, \citet{BBH15c}, \citet{BrBr17a}, and \citet{BLS22c} have studied efficiency, strategyproofness, and related properties with respect to \pc preferences.
When there are at least four objects, \pc preferences over lotteries can be cyclic even when preferences over objects are transitive. This phenomenon is known as the \emph{Steinhaus--Trybula paradox} \citep[see, e.g.,][]{StTr59a,Blyt72a,Pack82a,RuSe12a,BuPo18a}. 

While \pc preferences cannot be represented by a vNM utility function, they constitute a special case of SSB preferences. In fact, any asymmetric preference relation $\succ$ over $O$ can be conveniently represented by an SSB utility function $\phi$ whose entries are restricted to $\{-1,0,+1\}$ such that
\[ \phi(o, \hat o) = \begin{cases*}
	+1, & if $o \succ \hat o$ \\
	\phantom{+}0, & if $o \sim \hat o$ \\
	-1, & if $o \prec \hat o$
\end{cases*}. 
\tag{\pc extension}
\]
This representation can be viewed as the canonical SSB representation of ordinal preferences and yields preferences identical to the \pc relation. It follows from a more general observation by \citet[][Thm.~8]{Fish84d} that the \pc relation is a complete refinement of the \sd relation when preferences over objects are transitive \citep[see also][]{ABB14b}.
This implies that the efficiency notions based on the \pc relation, which we call \emph{\pc-efficiency}, are stronger than their counterparts based on the \sd relation, which we call \emph{\sd-efficiency}.\footnote{\sd-efficiency is sometimes also referred to as ordinal efficiency.}
On the other hand, \pc-envy-freeness, the envy-freeness notion obtained from the \pc relation, is weaker than strong \sd-envy-freeness (though it is stronger than weak \sd-envy-freeness). The logical relationships between these properties are visualized in \Cref{fig:axiom-hierarchy}.

\begin{figure}[htb]
\centering
\begin{tikzpicture}[
    >=Latex,
    implication/.style={->},%
    axiom/.style={font=\large, align=center}
]
    \node[axiom] (sd_ef) {strong \sd-envy-freeness};
    \node[axiom, below=of sd_ef] (pc_ef) {\pc-envy-freeness};
    \node[axiom, below=of pc_ef] (wsd_ef) {weak \sd-envy-freeness};

    \draw[implication] (sd_ef) -- (pc_ef);
    \draw[implication] (pc_ef) -- (wsd_ef);

    \node[axiom, left=3 of sd_ef] (pc_eff) {strong \pc-efficiency};
    \node[axiom, below= of pc_eff, xshift=-6em] (sd_eff) {strong \sd-efficiency};
    \node[axiom, below= of pc_eff, xshift=+6em] (wpc_eff) {weak \pc-efficiency};
	\node[axiom, below= of sd_eff, xshift=+6em] (wsd_eff) {weak \sd-efficiency};

    \draw[implication] (pc_eff) -- (sd_eff);
    \draw[implication] (pc_eff) -- (wpc_eff);
	\draw[implication] (sd_eff) -- (wsd_eff);
	\draw[implication] (wpc_eff) -- (wsd_eff);

\end{tikzpicture}
\caption{Logical relationships between efficiency and envy-freeness notions.}
\label{fig:axiom-hierarchy}
\end{figure}
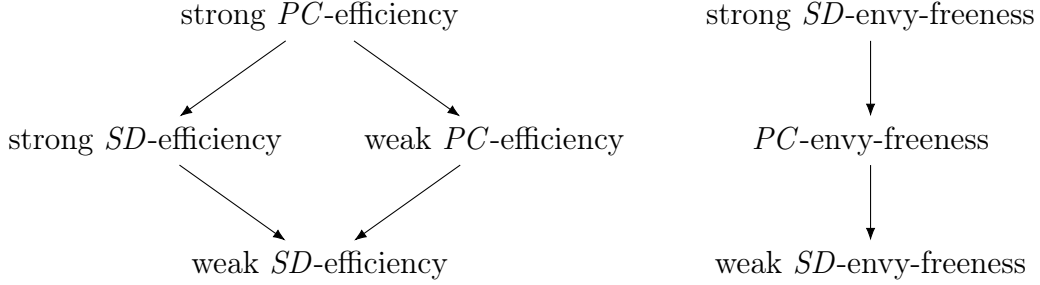

Three central ordinal random assignment rules are widely studied in the literature: random serial dictatorship, the probabilistic serial rule, and the popular random assignment rule.
We will introduce these rules and compare them to the extended version of the Hylland--Zeckhauser pseudo-market mechanism. The rules were originally defined for strict preference orders but later generalized to weak preference orders. Since our negative results hold even when preferences are strict, we focus on this preference domain.

For \emph{random serial dictatorship (\rsd)}, a picking order of agents is drawn uniformly at random, and the agents then successively choose their most preferred of the remaining objects in the drawn order \citep[see, e.g.,][]{AbSo98a}. \citet{BoMo01a} showed that \rsd violates \sd-efficiency but satisfies weak \sd-envy-freeness. The following example shows that \rsd violates the stronger property of \pc-envy-freeness. 

\begin{example}[\rsd violates \pc-envy-freeness]\label{ex:rsd_pceff}
	Consider the 7-agent preference profile with the following strict preferences over objects:
	\begin{align*}
		1&: a \succ b \succ c \succ d \succ e \succ f \succ g \\
		2&: b \succ a \succ c \succ d \succ e \succ f \succ g \\
		3, 4, 5, 6, 7&: a \succ c \succ d \succ e \succ f \succ g \succ b
	\end{align*}
	The random assignment returned by \rsd is
	\begin{equation*} X=
    	\begin{pmatrix}
        	\nicefrac{1}{6} & \nicefrac{5}{14} & \nicefrac{1}{21} & \nicefrac{1}{14} & \nicefrac{2}{21} & \nicefrac{5}{42} & \nicefrac{1}{7} \\
        	0 & \nicefrac{9}{14} & \nicefrac{1}{42} & \nicefrac{1}{21} & \nicefrac{1}{14} & \nicefrac{2}{21} & \nicefrac{5}{42} \\
        	\nicefrac{1}{6} & 0 & \nicefrac{13}{70} & \nicefrac{37}{210} & \nicefrac{1}{6} & \nicefrac{11}{70} & \nicefrac{31}{210} \\
			\nicefrac{1}{6} & 0 & \nicefrac{13}{70} & \nicefrac{37}{210} & \nicefrac{1}{6} & \nicefrac{11}{70} & \nicefrac{31}{210} \\
			\nicefrac{1}{6} & 0 & \nicefrac{13}{70} & \nicefrac{37}{210} & \nicefrac{1}{6} & \nicefrac{11}{70} & \nicefrac{31}{210} \\
			\nicefrac{1}{6} & 0 & \nicefrac{13}{70} & \nicefrac{37}{210} & \nicefrac{1}{6} & \nicefrac{11}{70} & \nicefrac{31}{210} \\
			\nicefrac{1}{6} & 0 & \nicefrac{13}{70} & \nicefrac{37}{210} & \nicefrac{1}{6} & \nicefrac{11}{70} & \nicefrac{31}{210}
		\end{pmatrix}.
	\end{equation*}
	It can be verified that agent 1 strictly \pc-prefers the lottery $x_2$ to $x_1$ with a margin of $\nicefrac{1}{1764}$.
\end{example}

The \emph{probabilistic serial (\ps)} rule was introduced by \citet{BoMo01a} and works by letting agents ``eat'' probability shares from their most preferred object at uniform speed subject to availability. \citeauthor{BoMo01a} showed that \ps satisfies strong \sd-envy-freeness and strong \sd-efficiency. While strong \sd-envy-freeness implies \pc-envy-freeness, the following example, adapted from \citet[][pp.~45--46]{Bran13b}, shows that \ps violates weak \pc-efficiency.\footnote{
While \ps is \sd-efficient, it is known that \ps fails to meet efficiency for certain vNM utility functions. For example, \Cref{ex:ps_pceff} can also be used to show that \ps is inefficient for equidistant vNM utility vectors $(3,2,1,0)$. It follows from a result by \citet{ChDo16a} that \ps is efficient for utility functions that are rapidly decreasing either at the top or at the bottom.
\citet{GTV26a} have shown that \ps is $(\ln n + 1)$-approximately efficient for vNM utilities. 
}

\begin{example}[\ps violates weak \pc-efficiency]\label{ex:ps_pceff}
	Consider the preference profile with $n=4$ where agents have the following strict orders over objects:
	\begin{align*}
		1, 2, 3&: a \succ b \succ c \succ d \\
		4&: b \succ a \succ c \succ d
	\end{align*}
	Now consider the two random assignments $X$ and $Y$. $X$ is the assignment returned by \ps, while $Y$ strongly \pc-dominates $X$.
	This means that every agent $i$ strictly prefers lottery $y_i$ to lottery $x_i$. 
	\begin{equation*} 
		X=
    	\begin{pmatrix}
        	\nicefrac{1}{3} & \nicefrac{1}{6} & \nicefrac{1}{4} & \nicefrac{1}{4} \\
        	\nicefrac{1}{3} & \nicefrac{1}{6} & \nicefrac{1}{4} & \nicefrac{1}{4} \\
        	\nicefrac{1}{3} & \nicefrac{1}{6} & \nicefrac{1}{4} & \nicefrac{1}{4} \\
			0 & \nicefrac{1}{2} & \nicefrac{1}{4} & \nicefrac{1}{4}
		\end{pmatrix},
		\qquad
		Y=
    	\begin{pmatrix}
        	\nicefrac{1}{3} & \nicefrac{1}{8} & \nicefrac{1}{3} & \nicefrac{5}{24} \\
        	\nicefrac{1}{3} & \nicefrac{1}{8} & \nicefrac{1}{3} & \nicefrac{5}{24} \\
        	\nicefrac{1}{3} & \nicefrac{1}{8} & \nicefrac{1}{3} & \nicefrac{5}{24} \\
			0 & \nicefrac{5}{8} & 0 & \nicefrac{3}{8}
		\end{pmatrix}.
	\end{equation*}
\end{example}
The above example is not exceptional. \citet{Mora22a} sampled thousands of preference profiles for varying $n$ and observed that the fraction of profiles at which \ps violates weak \pc-efficiency quickly approaches one as $n$ increases. \ps violates weak \pc-efficiency in more than one quarter of all profiles when $n=4$. 
The probability that a randomly sampled profile for $n=7$ has this property already exceeds 90\%. The \pc probability margin (i.e., the maximal difference of both sides of the \pc inequality above) also increases with $n$. When $n=5$, this margin can already exceed \nicefrac13. This means there are profiles where the \ps random assignment is dominated by another random assignment in which some agents are twice as likely to be better off. Weak \pc-efficiency failures of \rsd are even more frequent and more pronounced.

Finally, we consider \emph{popular random assignments (\pop)}, as proposed by \citet{KMN11a}. A random assignment is \emph{popular} if there does not exist another random assignment that is preferred by an \emph{expected} majority of agents. Popular random assignments always exist, but need not be unique. 
As pointed out by \citet{ABS13a}, popular random assignments are a special case of maximal lotteries, which exhibit desirable properties in social choice theory \citep[e.g.,][]{Fish84a,Bran13a,BrBr17a}. Maximal lotteries (and hence popular random assignments) maximize utilitarian SSB welfare for \pc utility functions \citep{BBH15c}. As a consequence, all popular random assignments are strongly \pc-efficient. 
However, \citet{ABS13a} constructed a preference profile that admits no popular random assignment that satisfies strong \sd-envy-freeness. \citet{BHS17a} extended this incompatibility to weak \sd-envy-freeness.
\citet{Mora22a} showed via computer simulations that \pc-envy-freeness violations of popular random assignments occur more frequently as $n$ increases and that \pc-envy probability margins can already exceed $\nicefrac13$ when $n=5$.

By contrast, the ordinal random assignment rule which returns a \emph{pseudo-market equilibrium with cost minimization},\footnote{In which agents with equal budgets select maximal affordable lotteries according to the \pc relation.} which we denote by \pchz, satisfies both \pc-efficiency and \pc-envy-freeness (see \Cref{section:SSB_with_strict_separation}).

To get more intuition into \pchz and how it differs from existing random assignment rules, consider the following example.

\begin{example}[\pchz ordinal random assignment rule]
	Let $N=\{1,2,3\}$, $\mO=\{a,b,c\}$ and consider the preference profile
	\[
		1:\ a\succ b\succ c,
		\qquad
		2:\ a\succ c\succ b,
		\qquad
		3:\ b\succ a\succ c.
	\]
	\ps and \rsd return the following random assignments for this profile.
	\[
		X^\ps=
		\begin{pmatrix}
		\nicefrac12 & \nicefrac14 & \nicefrac14\\
		\nicefrac12 & 0 & \nicefrac12\\
		0 & \nicefrac34 & \nicefrac14
		\end{pmatrix}\quad\text{and}\quad
		X^\rsd=
		\begin{pmatrix}
		\nicefrac12 & \nicefrac16 & \nicefrac13\\
		\nicefrac12 & 0 & \nicefrac12\\
		0 & \nicefrac56 & \nicefrac16
		\end{pmatrix}\text.
	\]

	By contrast, the unique \pchz pseudo-market equilibrium with cost minimization is given by the price vector $p^\ast=\left(\nicefrac23,\nicefrac13,0\right)$ together with the random assignment
	\[
		X^\pchz=
		\begin{pmatrix}
		\nicefrac12 & 0 & \nicefrac12\\
		\nicefrac12 & 0 & \nicefrac12\\
		0 & 1 & 0
		\end{pmatrix}.
	\]
	There are infinitely many popular random assignments in this profile, and $X^\pchz$ is one of them.

\end{example}

\Cref{table:axiom-violations} gives an overview of our results on ordinal random assignment. Note that \rsd and \ps already violate \emph{weak} \pc-efficiency. 

\begin{table}[htb]
	\centering
		\begin{tabular}{ccc}
			\toprule
			& \pc-efficiency & \pc-envy-freeness \\
			\midrule
			\rsd  & -- & -- \\
			\ps & -- & \checkmark \\  
			\pop & \checkmark & -- \\ 
		    \pchz & \checkmark & \checkmark \\
			\bottomrule
		\end{tabular}
	\caption{Properties of Random Assignment Rules}
	\label{table:axiom-violations}
\end{table}

\section{Conclusion and Discussion}

We extended the Hylland--Zeckhauser pseudo-market to continuous and convex preference relations. The resulting market equilibrium is weakly efficient and envy-free. Strong efficiency can be retained by restricting preferences to a natural subdomain of SSB preferences or by relaxing envy-freeness. These findings lead to a new ordinal random assignment rule that satisfies \pc-efficiency and \pc-envy-freeness.

As acknowledged by \citet{HyZe79a}, their pseudo-market is not strategyproof. Moreover, \citet{Zhou90a} proved that no efficient assignment rule that satisfies equal treatment of equals is strategyproof. \citet{AzBu19a} formalize in which sense HZ pseudo-markets are approximately strategyproof when there are many agents.

When comparing \pchz to other \emph{ordinal} random assignment rules, it has to be mentioned that \rsd is strongly \sd-strategyproof and \ps is weakly strategyproof. Note, however, that for weak preferences over objects, the canonical extension of \ps fails to meet weak \sd-strategyproofness. In fact, every \sd-efficient and strongly \sd-envy-free random assignment rule violates weak \sd-strategyproofness \citep{KaSe06a}.
We conjecture that no random assignment rule satisfies equal treatment of equals, \sd-efficiency, and weak \sd-strategyproofness for weak preferences.
For strict preferences, we conjecture that equal treatment of equals is incompatible with \pc-efficiency and \pc-strategyproofness.\footnote{Analogous statements were shown for the general social choice domain by \citet{BBEG16a} and \citet{BLS22c}.}

\citet{VaYa25a} constructed a vNM utility profile whose unique HZ equilibrium uses irrational probabilities. 
Furthermore, \citet{CCPY22a} show that computing approximate HZ equilibria is PPAD-hard once agents' vNM utilities take at least four distinct values. By contrast, \citet{YaLi26a} provide a polynomial-time algorithm for computing a constant-factor approximation by reducing the original market to a bi-valued instance.
\citet{TrVa26a} showed that approximately computing \emph{any} efficient and envy-free random assignment is PPAD-hard for vNM utilities. %
Since SSB preferences are more general than vNM preferences, these negative results carry over to SSB preferences.
An interesting open question is whether similar results can be shown for \pc preferences.

\subsection*{Acknowledgements}
This material is based on work supported by the Deutsche Forschungsgemeinschaft under grant {BR~2312/14-1}. We are grateful to Matthias Greger for supplying \Cref{ex:rsd_pceff}.

\newpage
\appendix

\section{Incompatibility of Efficiency and Envy-Freeness} \label{appendix:counter-example}

\CounterExample*
\begin{proof}
	Agents $1$ and $2$ have the strict transitive preference $a\succ b\succ c$ over objects, while agent $3$ satisfies $b\succ_3 a$, $a\succ_3 c$, and $b\sim_3 c$.

	Let $X$ be an envy-free random assignment, and write its bistochastic matrix as
	\begin{equation*}
	X=
	\begin{pmatrix}
	u & v & 1-u-v\\
	w & t & 1-w-t\\
	1-u-w & 1-v-t & u+v+w+t-1
	\end{pmatrix}.
	\end{equation*}

	By envy-freeness, we have, in particular,
	\begin{align}
		\phi_1(x_1,x_2)
		&= u+v-w-t+(ut-vw)=0, \label{eq:EF12}\\
		\phi_3(x_3,x_1)
		&=1-2v-w-2(ut-vw)\ge 0, \label{eq:EF31}\\
		\phi_3(x_3,x_2)
		&=1-u-2t+2(ut-vw)\ge 0. \label{eq:EF32}
	\end{align}
	The equality in \eqref{eq:EF12} follows because agents $1$ and $2$ have the same SSB preference: envy-freeness requires both $\phi_1(x_1,x_2)\ge 0$ and $\phi_2(x_2,x_1)\ge 0$, and these two quantities are negatives of one another.

	Using \eqref{eq:EF12}, we have
	\[
		ut-vw=-u-v+w+t.
	\]
	Substituting this into \eqref{eq:EF31} and \eqref{eq:EF32} gives
	\begin{equation*}
		3w+2t\le 1+2u \qquad\text{and}\qquad 3u+2v\le 1+2w.
	\end{equation*}
	Since $X$ is feasible, the entry $x_3(a)=1-u-w$ is non-negative, so $u+w\le 1$. Hence
	\begin{align*}
		5u+2v &\le 3,\\
		5w+2t &\le 3.
	\end{align*}
	Moreover, if both inequalities hold with equality, then the intermediate inequalities must also hold with equality, and therefore $u+w=1$.

	\medskip
	\noindent\textbf{Case 1:} $5u+2v<3$ or $5w+2t<3$.

	Consider the random assignment
	\begin{equation*}
		Y^1=
		\begin{pmatrix}
		\nicefrac{1}{2} & \nicefrac{1}{4} & \nicefrac{1}{4}\\
		\nicefrac{1}{2} & \nicefrac{1}{4} & \nicefrac{1}{4}\\
		0 & \nicefrac{1}{2} & \nicefrac{1}{2}
		\end{pmatrix}.
	\end{equation*}
	A direct computation gives
	\begin{align*}
		\phi_1(y^1_1,x_1) &=\nicefrac{1}{4}(3-5u-2v),\\
		\phi_2(y^1_2,x_2) &=\nicefrac{1}{4}(3-5w-2t),\\
		\phi_3(y^1_3,x_3) &=0.
	\end{align*}
	By the inequalities above, all three terms are non-negative. Since at least one of the two inequalities in this case is strict, at least one of the first two terms is strictly positive. Hence $Y^1$ Pareto dominates $X$.

	\medskip
	\noindent\textbf{Case 2:} $5u+2v=3$ and $5w+2t=3$.

	As noted above, equality implies $u+w=1$. Combining this with the two equalities gives
	\[
		v+t=\nicefrac{1}{2}.
	\]
	Hence
	\[
		x_3=(0,\nicefrac{1}{2},\nicefrac{1}{2}).
	\]
	Moreover, $u<1$, since otherwise $5u+2v=3$ would imply $v=-1$, contradicting feasibility. Similarly, $w<1$.

	Now consider the random assignment
	\begin{equation*}
		Y^2=
		\begin{pmatrix}
		\nicefrac{1}{2} & \nicefrac{1}{2} & 0\\
		\nicefrac{1}{2} & \nicefrac{1}{2} & 0\\
		0 & 0 & 1
		\end{pmatrix}.
	\end{equation*}
	Then
	\begin{align*}
		\phi_1(y^2_1,x_1) &=\nicefrac{1}{2}(2-3u-v) =\nicefrac{1}{4}(1-u),\\
		\phi_2(y^2_2,x_2) &=\nicefrac{1}{2}(2-3w-t) =\nicefrac{1}{4}(1-w),\\
		\phi_3(y^2_3,x_3) &=0,
	\end{align*}
	where the first two equalities use $5u+2v=3$ and $5w+2t=3$, respectively. Since $u<1$ and $w<1$, agents $1$ and $2$ are strictly better off under $Y^2$, while agent $3$ is indifferent. Thus $Y^2$ Pareto dominates $X$.

	In both cases, the envy-free assignment $X$ is Pareto dominated. Therefore, no random assignment is both efficient and envy-free for this profile.
\end{proof}

\section{A Generalized Pseudo-Market}
\label{appendix:pseudomarket}

\MarketExistence*

Recall from \Cref{section:pseudo-markets} the basic set-up of the pseudo-market, most of which we inherit from \citet{HyZe79a}.
In this section, we give more detailed definitions of the price set $P$, the demand correspondence $D_i$ for each agent $i \in N$, and a price update mechanism that responds to excess demand and ensures market clearing at a fixed point.
Along the way, we collect auxiliary lemmas that show that, even for general continuous and convex preferences, these three components together satisfy the conditions necessary for Kakutani's fixed-point theorem.

\paragraph{Price set.}
The price set $P$, formally defined by
\[
	P \coloneqq
	\left\{
	p\in\R^{n}_{\ge 0}\colon \min_{o\in\mO}p_o=0 \text{ and } \left\|p-\frac{p\cdot\mathbf 1}{n}\mathbf 1\right\|_1\le 2
	\right\},
\]
is not convex.
Therefore, we follow \citet{HyZe79a} in obtaining $P$ via a homeomorphic normalization map from a compact and convex parameter space $S$.
Let
\begin{equation*}
	S \coloneqq\left\{s \in \R^{n} \colon \sum_{o\in \mO} s_o = 0 \text{ and } \|s\|_1 \le 2 \right\}.
\end{equation*}
Define the normalization map $f \colon S \to \R^{n}_{\ge 0}$ by
\begin{equation*}
	f(s) \coloneqq s - \min_{o\in\mO}s_o \, \mathbf 1.
\end{equation*}
With this, we obtain $P = f(S)$, noting that $f$ is a homeomorphism.
A key identity, which we use in the proof of \Cref{theorem:market_existence}, is that for every \emph{zero-sum} vector $z\in \R^{n}$,
\begin{equation} \label{eq:price-identity}
	f(s)\cdot z = \bigl(s - \min_o s_o \, \mathbf 1 \bigr) \cdot z = s \cdot z.
\end{equation}

\paragraph{Demand correspondences.}
Recall that each agent is given a virtual budget $b_i > 0$ satisfying $\sum_{i \in N} b_i = 1$. The budget set of agent $i$ with respect to a price vector $p \in P$ is therefore
\begin{equation*}
	B_i(p) \coloneqq\{x\in \Delta \colon p \cdot x \le b_i\}.
\end{equation*}

\begin{lemma}
	For every agent $i \in N$ and every price vector $p\in P$, the set $B_i(p)$ is non-empty, compact, and convex.
\end{lemma}
\begin{proof}
	Since $\min_{o\in \mO} p_o=0$, there exists some object $o^\ast$ with $p_{o^\ast}=0$. 
	Hence the degenerate lottery $e_{o^\ast}\in \Delta$ satisfies
	\begin{equation*}
		p \cdot e_{o^\ast} = 0 \le b_i,
	\end{equation*}
	so $B_i(p)$ is non-empty.

	It is compact because it is the intersection of the compact set $\Delta$ with the closed halfspace $\{x: p\cdot x\le b_i\}$. 
	It is convex because both these sets are convex.
\end{proof}

\begin{lemma}
	For every agent $i \in N$, the correspondence $B_i:P\rightrightarrows \Delta$ is continuous.
\end{lemma}
\begin{proof}
	Since $\Delta$ is compact, upper hemicontinuity is equivalent to the graph
	\begin{equation*}
		\{(p,x)\in P \times \Delta \colon p \cdot x \le b_i\}
	\end{equation*}
	being closed. 
	This follows from the continuity of the inner product $(p,x) \mapsto p \cdot x$.

	For lower hemicontinuity, fix $p \in P$, let $V \subseteq \Delta$ be an open subset, and suppose $B_i(p) \cap V \neq \emptyset$. 
	Pick $x \in B_i(p) \cap V$. 
	Choose an object $o^\ast$ with $p_{o^\ast}=0$, and define $\bar x:=e_{o^\ast}$. 
	Then $p\cdot \bar x = 0 < b_i$.
	Choose $\lambda\in(0,1)$ small enough that
	\begin{equation*}
		x_\lambda \coloneqq(1-\lambda)x + \lambda \bar x \in V.
	\end{equation*}
	Since $\Delta$ is convex, $x_\lambda\in \Delta$. 
	Moreover,
	\begin{equation*}
		p\cdot x_\lambda = (1-\lambda)(p\cdot x)+\lambda(p\cdot \bar x) \le (1-\lambda)b_i < b_i.
	\end{equation*}
	By continuity of $p'\mapsto p'\cdot x_\lambda$, there exists a neighborhood $U \subseteq P$ of $p$ such that
	\begin{equation*}
		p'\cdot x_\lambda < b_i \qquad\text{for all } p'\in U.
	\end{equation*}
	Hence $x_\lambda\in B_i(p')\cap V$ for all $p'\in U$. This proves lower hemicontinuity.
\end{proof}

Next, \Cref{lemma:demand-compact-convex,lemma:maximal_set_upper_hc} show that the demand correspondence defined as
\begin{equation*}
	D_i(p) \coloneqq\{x \in B_i(p) \colon x \succsim_i y \text{ for all } y \in B_i(p)\}
\end{equation*}
is non-empty, compact, convex and satisfies upper hemicontinuity.

\begin{lemma} \label{lemma:demand-compact-convex}
	For every agent $i \in N$ and every price vector $p\in P$, the set $D_i(p)$ is non-empty, compact, and convex.
\end{lemma}
\begin{proof}
	Since $B_i(p)$ is non-empty, compact and convex, and the preference relation $\succ_i$ has closed weak upper contour sets and convex strict upper contour sets, the existence of maximal elements follows from \citet{Sonn71a}.
	Hence, $D_i(p)\neq\emptyset$.

	To prove compactness, it is enough to show that $D_i(p)$ is closed in $B_i(p)$. 
	Let $(x^k)$ be a convergent sequence in $D_i(p)$ with $x^k\to x\in B_i(p)$ as $k \to \infty$. 
	If $x \notin D_i(p)$, then there exists $y\in B_i(p)$ such that $y\succ_i x$.
	Since the strict lower contour set of $y$ is open, we have $y\succ_i x^k$ for all sufficiently large $k$, contradicting $x^k\in D_i(p)$. 
	Hence, $x\in D_i(p)$, so $D_i(p)$ is closed and therefore compact.

	To prove convexity, let $x,y\in D_i(p)$ and $\lambda\in[0,1]$. 
	Set $z \coloneqq \lambda x+(1-\lambda)y$.
	Since $B_i(p)$ is convex, we have $z\in B_i(p)$. 
	Suppose, for contradiction, that there exists $z'\in B_i(p)$ such that $z'\succ_i z$.
	Because $x$ and $y$ are maximal and $\succsim_i$ is complete, we must have $x \succsim_i z'$ and $y \succsim_i z'$, i.e., $x$ and $y$ are in the weak upper contour set $W_i(z')$.
	Since $W_i(z')$ is convex by assumption, this implies $z\in W_i(z')$ and $z\succsim_i z'$, contradicting $z'\succ_i z$.
	Hence, $z\in D_i(p)$, and $D_i(p)$ is convex.
\end{proof}

\begin{lemma}\label{lemma:maximal_set_upper_hc}
	For every agent $i \in N$, the correspondence $D_i \colon P \rightrightarrows \Delta$ is upper hemicontinuous.
\end{lemma}

\begin{proof}
	Let $(p^k)$ be a convergent sequence in $P$ with $p^k \to p$ as $k \to \infty$, let $x^k\in D_i(p^k)$ be a maximal affordable lottery for every $k$, and suppose that $x^k\to x \in \Delta$ as $k \to \infty$. 
	Since the budget-set correspondence $B_i$ is upper hemicontinuous and $x^k\in B_i(p^k)$ for every $k$, we have $x\in B_i(p)$.

	Suppose, towards a contradiction, that $x\notin D_i(p)$. By completeness of $\succsim_i$, there exists some $y\in B_i(p)$ such that $y\succ_i x$.
	Since $B_i$ is lower hemicontinuous, there exists a convergent sequence $y^k\in B_i(p^k)$ such that $y^k\to y$ as $k \to \infty$.

	We now use that the strict preference graph $G(\succ_i)$ is open.
	Since $(y,x)\in G(\succ_i)$ and $(y^k,x^k)\to(y,x)$, it follows that $(y^k,x^k)\in G(\succ_i)$ for all sufficiently large $k$. Hence, $y^k\succ_i x^k$ for all sufficiently large $k$.
	This contradicts $x^k\in D_i(p^k)$, because $y^k\in B_i(p^k)$ and $y^k \succ_i x^k$. 
	Therefore, $x\in D_i(p)$, and $D_i(\cdot)$ has a closed graph. Since $\Delta$ is compact, this implies that $D_i(\cdot)$ is upper hemicontinuous.
\end{proof}

We remark that \Cref{lemma:maximal_set_upper_hc} mirrors \citet[][Thm.~2]{Walk79a}.
From each agent's demand correspondence $D_i(\cdot)$, we now define the \emph{aggregate-demand correspondence} $D: P \rightrightarrows \Delta^n$ simply by the Cartesian product
\begin{equation*}
	D(p) \coloneqq \prod_{i\in N} D_i(p).
\end{equation*}
Since each $D_i(\cdot)$ is non-empty, compact-valued, convex-valued, and upper hemicontinuous, the same is true for $D$.
An element $X=(x_i)_{i \in N}$ of $D(p)$ is exactly a selection of lottery $x_i$ for each agent $i$, under the assumption that each agent selects an element of her demand set $D_i(p)$, i.e., a maximal element of the budget set $B_i(p)$.

\paragraph{Price updates.}
Recall that the price set $P = f(S)$ is obtained via a homeomorphic normalization map from a compact convex parameter space $S$.
We now define the price-update dynamics on $S$, by which the market adjusts prices in response to excess demand in a selection $X \in D(p)$.
Given any selection $X\in \Delta^n$ of lotteries, the excess demand (compared to the unit supply of each object) is denoted by
\begin{equation*}
	z(X) \coloneqq \sum_{i\in N} x_i - \mathbf 1.
\end{equation*}
Using this notation, we define the price-update correspondence $\Pi: \Delta^n \rightrightarrows S$ by
\begin{equation*}
	\Pi(X) \coloneqq \arg\max_{s\in S}~ s \cdot z(X).
\end{equation*}
In words, $\Pi(X)$ selects those price parameters $s\in S$ that maximize the dot product with the excess-demand vector induced by the current lottery selection $X$.

\begin{lemma} \label{lemma:price-update-correspondence}
	The correspondence $\Pi:\Delta^n\rightrightarrows S$ is non-empty, compact-valued, convex-valued, and upper hemicontinuous.
\end{lemma}
\begin{proof}
	For a fixed selection $X \in \Delta^n$, the map $s\mapsto s\cdot z(X)$ is continuous and linear on the compact convex parameter space $S$. 
	Hence, the argmax set of this map is non-empty, compact, and convex. 
	Since $X \mapsto z(X)$ is also continuous, upper hemicontinuity follows from Berge's maximum theorem \citep{Berg63a}.
\end{proof}

We can now give the proof of \Cref{theorem:market_existence}.

\begin{proof}[Proof of \Cref{theorem:market_existence}]
	Define the correspondence $\Gamma: \Delta^n \times S \rightrightarrows \Delta^n \times S$ by
	\begin{equation*}
		\Gamma(X,s) \coloneqq D(f(s)) \times \Pi(X).
	\end{equation*}
	The domain $\Delta^n \times S$ is non-empty, compact, and convex. 
	By \Cref{lemma:demand-compact-convex,lemma:maximal_set_upper_hc,lemma:price-update-correspondence}, $\Gamma$ is non-empty, compact-valued, convex-valued, and upper hemicontinuous since $D(\cdot)$ and $\Pi(\cdot)$ both satisfy each of these properties. 
	Kakutani's fixed-point theorem therefore yields a fixed point $(X^\ast,s^\ast)$ such that
	\begin{equation*}
		X^\ast\in D(f(s^\ast)) \qquad\text{and}\qquad s^\ast\in \Pi(X^\ast).
	\end{equation*}
	We claim that the pair $(X^\ast, p^\ast)$ with $p^\ast \coloneqq f(s^\ast)$ is a pseudo-market equilibrium. 
	By definition, $X^\ast \in D(p^\ast)$ implies that each $x^\ast_i \in D_i(p^\ast)$, i.e., rationality is satisfied.
	It remains to show that the market-clearing condition is satisfied.
	
	Let $z^\ast \coloneqq z(X^\ast)$.
	Since each $x_i^\ast\in \Delta$, we get
	\begin{equation*}
		\sum_{o\in \mO} z_o^\ast = \sum_{o\in \mO}\left(\sum_{i\in N} x_i^\ast(o) -1\right) = \sum_{o \in \mO}\sum_{i \in N} x^\ast_i(o) - n = 0.
	\end{equation*}

	Because each $x_i^\ast \in B_i(p^\ast)$, we have $p^\ast\cdot x_i^\ast \le b_i$ for every agent $i$.
	Summing over all agents gives
	\begin{equation*}
		p^\ast\cdot \sum_{i\in N} x_i^\ast \le \sum_{i\in N} b_i =1.
	\end{equation*}
	Hence, $p^\ast\cdot z^\ast \le 1-p^\ast\cdot \mathbf 1$.

	We now show that $z^\ast = \mathbf 0$, i.e., the market clears.
	To this end, assume for contradiction that $z^\ast\neq \mathbf 0$.
	On the one hand, since $z^\ast$ is zero-sum and $S$ is the closed $\ell_1$-ball of radius $2$ in the zero-sum hyperplane, the linear functional $s \mapsto s\cdot z^\ast$ attains a strictly positive maximum on $S$.
	Because $s^\ast \in \Pi(X^\ast)$, we have $s^\ast\cdot z^\ast >0$.
	By the identity \eqref{eq:price-identity} for zero-sum vectors,
	\begin{equation*}
		p^\ast\cdot z^\ast = s^\ast\cdot z^\ast >0.
	\end{equation*}
	On the other hand, since $z^\ast\neq 0$, the maximizer $s^\ast$ lies on the boundary of $S$, so
	$\|s^\ast\|_1=2$. Because $s^\ast\cdot\mathbf 1=0$, the total negative mass of
	$s^\ast$ equals $1$. Hence
	\[
		-\min_{o\in\mO}s_o^\ast\ge \frac{1}{n}.
	\]
	Therefore,
	\[
		p^\ast\cdot\mathbf 1
		= \sum_{o\in\mO}\bigl(s_o^\ast-\min_{\hat o\in\mO}s_{\hat o}^\ast\bigr)
		= -n\min_{\hat o\in\mO}s_{\hat o}^\ast \ge 1.
	\]
	On the other hand, affordability gives
	\[
		p^\ast\cdot z^\ast \le 1-p^\ast\cdot\mathbf 1 \le 0,
	\]
	which contradicts $p^\ast\cdot z^\ast>0$. Hence, $z^\ast=\mathbf 0$.
\end{proof}

\subsection{SSB Preferences with Strict Maximality} \label{appendix:strictness_assumption}

We first provide the proof of \Cref{proposition:characterization_SSB_strict_maximality}, which characterizes SSB preferences satisfying strict maximality via its matrix representation.

\strictMaximalityCharacterization*
\begin{proof}
	Given a set $E \subseteq \Delta$ of degenerate lotteries, we denote by $\operatorname{conv} E \subseteq \Delta$ the face of $\Delta$ spanned by the vertices in $E$.
	Given a lottery $x \in \Delta$, we denote by $\supp(x) \subseteq \mO$ the support of $x$.

	First suppose that the preference relation $\succ_i$ represented by $\phi_i$ satisfies strict maximality, and let $M_i$ be the set of maximal elements of $\Delta$ under $\succ_i$. 
	Fix some $x\in M_i$. 
	For all $y \in \Delta$, by strict maximality we have $\phi_i(x, y) \geq 0$ with $\phi_i(x, y) = 0$ if and only if $y \in M_i$.
	Therefore,
	\[
		M_i = \arg\min_{y \in \Delta} \phi_i(x,y).
	\]
	Since $y\mapsto \phi_i(x,y)$ is linear and non-negative on the simplex $\Delta$, the set $M_i$ which minimizes this function must be a face of $\Delta$. 
	Hence, there exists a non-empty set $T\subseteq O$ of objects such that
	\[
		M_i=\operatorname{conv}\{e_o\colon o\in T\}.
	\]
	We now derive the matrix form by reordering the objects in $O$ so that objects in $T$ precede those in $O \setminus T$. 
	For any two objects $o,o'\in T$, $e_o,e_{o'}\in M_i$, so the preceding argument implies $\phi_i(e_o,e_{o'})=0$.
	Thus, the block of $\phi_i$ corresponding to objects in $T$ is the zero matrix. 
	If $o\in T$ and $o'\notin T$, then $e_o\in M_i$ while $e_{o'}\notin M_i$. By strict maximality, $\phi_i(e_o,e_{o'})>0$.
	Therefore, after reordering, the SSB matrix $\phi_i$ has the form
	\[
		\phi_i=
		\begin{pmatrix}
		0 & Q\\
		-Q^\top & C
		\end{pmatrix},
	\]
	where every entry of $Q$ is strictly positive. The lower-right block $C$ is arbitrary but skew-symmetric, because $\phi_i$ is skew-symmetric.

	Conversely, suppose that the SSB matrix $\phi_i$ has the form
	\[
		\phi_i=
		\begin{pmatrix}
		0 & Q\\
		-Q^\top & C
		\end{pmatrix},
	\]
	where every entry of $Q$ is strictly positive, and let $T \subseteq O$ be the set of objects corresponding to the top-left zero block.
	
	Let $x \in \operatorname{conv}\{e_o:o\in T\}$. We claim that for all $y \in \Delta$, we have $x \succsim_i y$ with $x \sim_i y$ if and only if $y \in \operatorname{conv}\{e_o:o\in T\}$.
	Note that since $\supp(x) \subseteq T$, for all $o \in \supp(x)$ and $\hat{o} \in \supp(y)$, we have $\phi_i(e_o, e_{\hat o}) \geq 0$ with $\phi_i(e_o, e_{\hat o}) = 0$ if and only if $\hat{o} \in T$.
	Hence, the expansion
	\begin{align*}
		\phi_i(x,y) &= \sum_{o\in \supp(x)}\sum_{\hat o \in \supp(y)}x(o) y(\hat o)\phi_i(e_o,e_{\hat o})
	\end{align*}
	gives $\phi_i(x, y) \geq 0$ with $\phi_i(x, y) = 0$ if and only if $\supp(y) \subseteq T$, i.e., $y \in \operatorname{conv}\{e_o:o\in T\}$.

	Therefore, we may conclude that $M_i = \operatorname{conv}\{e_o:o\in T\}$, and $x \succ_i y$ for all $x \in M_i$ and $y \notin M_i$. 
	That is, strict maximality is satisfied.
\end{proof}

We now come to the proof of \Cref{theorem:existence-with-cost-minimization}: assuming SSB preferences, pseudo-market equilibria with cost-minimization always exist.

\CostMinimizationExistence*

Most of the proof is identical to that of \Cref{theorem:market_existence}. 
Specifically, the price set $P$ and price-update correspondence $\Pi$ are defined identically. 
The only difference is that we refine the demand correspondences $D_i$ by enforcing cost-minimization. 
Therefore, it suffices to prove that the refined correspondence $\widehat{D}_i(\cdot)$ is upper hemicontinuous for each agent $i \in N$.

\begin{lemma}\label{lemma:ssb_budget_spending}
	Suppose the preference relation $\succ_i$ is represented by an SSB utility function $\phi_i$, and let $p \in P$ be any price vector.
	If $y\in D_i(p)$ and $p\cdot y<b_i$, then $y \in M_i$.
\end{lemma}
\begin{proof}
	Suppose, towards a contradiction, that there exists $z\in\Delta$ such that $z \succ_i y$.
	For $\varepsilon\in(0,1)$, define $y^\varepsilon \coloneqq (1-\varepsilon)y+\varepsilon z$.
	By bilinearity of $\phi_i$,
	\[
		\phi_i(y^\varepsilon,y)
		= (1-\varepsilon)\phi_i(y,y)+\varepsilon\phi_i(z,y)
		= \varepsilon\phi_i(z,y)>0.
	\]
	Hence $y^\varepsilon\succ_i y$. 
	Since $p\cdot y<b_i$, we can choose $\varepsilon>0$ sufficiently small so that $p\cdot y^\varepsilon<b_i$.
	Thus $y^\varepsilon\in B_i(p)$ and $y^\varepsilon\succ_i y$, contradicting $y\in D_i(p)$. Therefore, $y \in M_i$.
\end{proof}

Note that for all $p \in P$, the non-emptiness, compactness and convexity of $\widehat{D}_i(p)$ carry over immediately from the corresponding properties of $D_i(p)$, since $\widehat D_i(p)$ is the argmin set of the continuous linear function $x\mapsto p\cdot x$ on the non-empty compact convex set $D_i(p)$.

\begin{lemma}
	Suppose the preference relation $\succ_i$ is represented by an SSB utility function $\phi_i$. The correspondence
	\[
		\widehat D_i(p) \coloneqq \arg\min_{x\in D_i(p)}p\cdot x
	\]	
	is upper hemicontinuous.
\end{lemma}
\begin{proof}
	Let $(p^k)$ be a convergent sequence in $P$ with $p^k \to p$ as $k \to \infty$, $x^k \in \widehat D_i(p^k)$ for each $k$, and suppose $x^k \to x$ as $k \to \infty$. 
	Recall that the correspondence $D_i(\cdot)$ is upper hemicontinuous by \Cref{lemma:maximal_set_upper_hc}.
	Since $x^k \in \widehat D_i(p^k)\subseteq D_i(p^k)$ and $D_i(\cdot)$ is upper hemicontinuous, we have $x\in D_i(p)$.

	We show that $x\in\widehat D_i(p)$. 
	Suppose, towards a contradiction, that there exists $y\in D_i(p)$ with $p\cdot y<p\cdot x$.
	Since $x\in B_i(p)$, we have $p\cdot x\le b_i$, and therefore $p\cdot y<b_i$.
	By \Cref{lemma:ssb_budget_spending}, $y \in M_i$, i.e., $y$ is maximal with respect to $\succ_i$ in the full simplex $\Delta$. 
	Hence, $y \in D_i(p')$ for every price vector $p' \in P$ with $y \in B_i(p')$.

	By continuity of the inner product, we have as $k \to \infty$,
	\[
		p^k\cdot y\to p\cdot y
		\qquad \text{and} \qquad
		p^k\cdot x^k\to p\cdot x.
	\]
	Since $p\cdot y<p\cdot x \leq b_i$, we have $p^k\cdot y<p^k\cdot x^k \leq b_i$, i.e., $y\in B_i(p^k)$, for sufficiently large $k$.
	Since $y \in M_i$, we also have $y\in D_i(p^k)$ for such $k$. 
	But $x^k\in\widehat D_i(p^k)$, so $x^k$ is a cost-minimizing element of $D_i(p^k)$. 
	This contradicts
	\[
		y\in D_i(p^k)
		\qquad \text{and} \qquad
		p^k\cdot y<p^k\cdot x^k.
	\]
	Hence, $x\in\widehat D_i(p)$ and the graph of $\widehat{D}_i(\cdot)$ is closed. 
	Since $\Delta$ is compact, $\widehat{D}_i(\cdot)$ is upper hemicontinuous.
\end{proof}

\begin{proof}[Proof of \Cref{theorem:existence-with-cost-minimization}]
	Define the correspondence $\widehat{\Gamma}: \Delta^n \times S \rightrightarrows \Delta^n \times S$ by 
	\[
		\widehat{\Gamma}(X, s) \coloneqq \widehat{D}(f(s)) \times \Pi(X),
	\]
	where $\widehat{D} \colon P \rightrightarrows \Delta^n$ is the aggregate demand correspondence given by
	\[\widehat{D}(p) \coloneqq \prod_{i \in N} \widehat{D}_i(p).\]
	Since $\widehat{D}_i(\cdot)$ is non-empty, compact-valued, convex-valued and upper hemicontinuous, so is $\widehat D$ and therefore also $\widehat \Gamma$.
	By Kakutani's fixed-point theorem, there exists a fixed point $(X^\ast, s^\ast)$ of $\widehat \Gamma$.
	Since $X^\ast \in \widehat{D}(f(s^\ast)) \subseteq D(f(s^\ast))$, the same argument as in the proof of \Cref{theorem:market_existence} implies that $(X^\ast, f(s^\ast))$ is a pseudo-market equilibrium.
	By definition, $X^\ast \in \widehat{D}(f(s^\ast))$ means that this pseudo-market equilibrium also satisfies cost-minimization.
\end{proof}

\begin{example}[Pseudo-market equilibrium with cost-minimization need not exist in general]
	\label{example:non-existence-cost-minimization}
	Let $O=\{a,b\}$, $n = 2$ with equal budgets $b_1 = b_2 = \nicefrac12$, and let both agents have identical preferences represented by the following utility function.
	Given $x \in \Delta$, define \[u(x) = \max \{0, x(a) - \nicefrac{1}{2}\}.\]
	That is, each agent is indifferent between all lotteries with $x(a) \leq \nicefrac12$, and when $x(a) > \nicefrac12$, has utility strictly increasing in $x(a)$.
	
	We claim that there exists no pseudo-market equilibrium with cost-minimization.
	For $n = 2$, the price set $P$ is given by \[P = \{p \in \R^n \colon p_a = 0 \text{ and } p_b \in [0, 2]\} \cup \{p \in \R^n \colon p_b = 0 \text{ and } p_a \in [0, 2]\}.\]

	Note that for any price vector $p \in P$ with $p_a = 0$, the lottery $e_a$ is in $B_i(p)$ and, hence, is the unique maximal affordable lottery for both agents $i \in \{1, 2\}$.
	Both agents then select $x^\ast_i = e_a$, and the market-clearing condition fails.
	The same is true if $p_b = 0$ and $p_a \in [0, \nicefrac12]$.
	The only case remaining is where $p_b = 0$ and $p_a \in (\nicefrac12, 2]$; note that in this case, each budget set $B_i(p)$ is given by $\left\{x \in \Delta: x(a) \leq \frac{1}{2p_a}\right\}$.
	If $p_a \in (\nicefrac12, 1)$, then $\frac{1}{2p_a} \in (\nicefrac12, 1)$, and hence $x^\ast_i = \left(\frac{1}{2p_a}, 1 - \frac{1}{2p_a}\right)$ is the unique maximal affordable lottery for each agent.
	Once again, both agents selecting the same lottery violates market-clearing.
	On the other hand, if $p_a \in [1, 2]$, then $\frac{1}{2p_a} \in [\nicefrac14, \nicefrac12]$.
	Then, each agent $i$ is indifferent among all lotteries in $B_i(p)$.
	Assuming cost-minimization, both agents select the uniquely cheapest lottery $e_b$, once again violating the market-clearing condition.

	Hence, no pseudo-market equilibrium with cost-minimization exists for this preference profile.
	Note, for example, that the assignment where both agents receive the lottery $x^\ast_i = (\nicefrac12, \nicefrac12)$ with price vector $p = (1, 0)$, forms a pseudo-market equilibrium but violates cost-minimization.
\end{example}

\section{Efficiency and $\alpha$-Envy-Freeness for SSB Preferences} \label{appendix:alpha-envy-free}

\FixedPointExistence*
\begin{proof}
	We apply Kakutani's fixed-point theorem.
	First note that the domain $\mathcal{M} \times \Omega_\varepsilon$ is non-empty, compact and convex.
	The second component $\{g(X, \omega)\}$ is by definition a singleton set, so in particular, non-empty and convex for all $X$ and $\omega$.
	Moreover, we claim that the map $(X, \omega) \mapsto g(X, \omega)$ is continuous and hence has a closed graph.
	To see this, note that for each pair of agents $i, j \in N$, the map $X \mapsto \phi_i(x_j, x_i)$ is continuous, and hence so is $\nu_i$, as the maximum of finitely many continuous functions.
	Therefore, $\mu(X, \omega)$ is continuous in both $X$ and $\omega$, and so is $g(X, \omega)$, which simply applies the continuous projection $\operatorname{proj}_{\Omega_\varepsilon}$.
	This means that for the Kakutani conditions to hold, it suffices to show that for all $\omega \in \Omega_\varepsilon$, the set $F(\omega)$ is non-empty, convex, and that the graph of $\omega \mapsto F(\omega)$ is closed.

	Consider the symmetric two-player zero-sum game in which the pure strategies of both players are the deterministic assignments, so that mixed strategies correspond to random assignments in $\mathcal{M}$, and the payoffs are given by the weighted aggregate SSB function $\Phi_\omega(X, Y) \coloneqq \sum_{i \in N} \omega_i \phi_i(x_i, y_i)$.
	Then, the set $F(\omega)$ corresponds exactly to the set of mixed maximin strategies of this game, and as such is non-empty and convex by the minimax theorem.
	To see that the graph $\{(\omega, F(\omega))\}$ is closed, let $(\omega^k)$ and $(X^k)$ be convergent sequences in $\Omega_\varepsilon$ and $\mathcal{M}$ respectively, with $\omega^k \to \omega$, $X^k \to X$, and $X^k \in F(\omega^k)$ for all $k$.
	For each $Y \in \mathcal{M}$ and each $k$, we have $\sum_{i \in N} \omega^k_i \phi_i(x^k_i, y_i) \geq 0$.
	Since $\Phi_\omega(X, Y)$ is continuous in both $X$ and $\omega$, we have $\sum_{i \in N} \omega_i \phi_i(x_i, y_i) = \lim_{k \to \infty} \sum_{i \in N} \omega^k_i \phi_i(x^k_i, y_i) \geq 0$.
	Hence, $X \in F(\omega)$, and $\omega \mapsto F(\omega)$ has a closed graph.

	We have proved that the correspondence $\Gamma$ satisfies the conditions of the Kakutani fixed-point theorem. 
	Therefore, there exists a fixed point $(X^\ast, \omega^\ast)$ as required.
\end{proof}

\FixedPointEnvyFreeness*

Towards proving \Cref{proposition:fixed-point-envy-freeness}, first define the following instance-dependent parameter $\sigma$, which is the largest SSB matrix entry taken over all agents.
\begin{equation*}
	\sigma \coloneqq \max_{i \in N, o, \hat o \in O} \phi_i(o, \hat o),
\end{equation*}
and let
\begin{equation*}
	\rho = \min\left\{\nicefrac{1}{2}, \nicefrac{\alpha}{\sigma}\right\}.
\end{equation*}
Note that $\sigma > 0$ unless every agent is completely indifferent between all lotteries, in which case every random assignment is envy-free and \Cref{theorem:alpha-envy-free} holds trivially; we therefore assume $\sigma > 0$ in the following.

\begin{lemma} \label{lemma:weight-ratio}
	Let $X \in F(\omega)$ for some $\omega \in \Omega_\varepsilon$. 
	For any two agents $i, j \in N$, $\phi_i(x_j, x_i) > \alpha$ implies 
	\begin{equation*}
		\omega_j > \rho \omega_i.
	\end{equation*}
	In other words, $\omega_j \leq \rho \omega_i$ implies $\phi_i(x_j, x_i) \leq \alpha$.
\end{lemma}
\begin{proof}
	Recall that $\rho$ is defined as $\min\{1/2, \alpha/\sigma\}$.
	We will prove the claimed inequality in the case where $\rho = \alpha / \sigma$.
	The case where $\alpha / \sigma > 1/2$ and $\rho = 1/2$ then follows from the same argument, since $\rho \leq \nicefrac{\alpha}{\sigma}$.
	The case where $i = j$ is also immediate, so we assume $i \neq j$.

	Let $X \in F(\omega)$, $i \neq j$, and assume $\phi_i(x_j, x_i) > \alpha$.
	Let $X^{ij}$ be the random assignment obtained from $X$ by swapping the rows corresponding to agents $i$ and $j$.
	Since $X \in F(\omega)$, we have
	\begin{equation} \label{eq:welfare-maximizing}
		\sum_{k \in N} \omega_k \phi_k(x_k, x^{ij}_k) \geq 0.
	\end{equation}
	Since $\phi_k(x_k, x^{ij}_k)=0$ for all $k \notin \{i, j\}$ and $x^{ij}_i = x_j$ and $x^{ij}_j = x_i$, \Cref{eq:welfare-maximizing} reduces to
	\begin{equation*}
		\omega_i \phi_i(x_i, x_j) + \omega_j \phi_j(x_j, x_i) \geq 0,
	\end{equation*}
	and hence
	\begin{equation*}
		\omega_j \geq \nicefrac{\phi_i(x_j, x_i)}{\phi_j(x_j, x_i)} \omega_i.
	\end{equation*}
	Note that we may divide by $\phi_j(x_j, x_i)$: if $\phi_j(x_j, x_i)$ were negative, the left-hand side of the preceding inequality would be negative, since $\phi_i(x_i, x_j) = -\phi_i(x_j, x_i) < 0$; and if $\phi_j(x_j, x_i) = 0$, then $X^{ij}$ would Pareto dominate $X$, contradicting $X \in F(\omega)$.
	Moreover, since $\phi_i(x_j, x_i) > \alpha$ and $\phi_j(x_j, x_i) \leq \sigma$, we get
	\begin{align*}
        \omega_j > \nicefrac{\alpha}{\sigma} \omega_i = \rho \omega_i
    \end{align*}
	as required.
\end{proof}

We finally set the value of $\varepsilon$, the lower bound on the weights $\omega_i$ in the definition of $\Omega_\varepsilon$.
This value is given by
\begin{equation*}
	\varepsilon \coloneqq \nicefrac{\rho^{n+1}}{n}.
\end{equation*}

Given any random assignment $X$, define the directed \emph{$\alpha$-envy graph} of $X$ to have the set $N$ of agents as its vertex set, and to include a directed edge from agent $i$ to agent $j$ if and only if $\phi_i(x_j, x_i) > \alpha$.
A sink in this graph (a vertex with no outgoing edge) corresponds to an agent who does not envy any other agent by more than $\alpha$.

\begin{lemma} \label{lemma:sink-existence}
	Let $X \in F(\omega)$ for some $\omega \in \Omega_{\varepsilon}$. Then, there exists at least one sink $i^\ast$ in the $\alpha$-envy graph of $X$ which satisfies $\omega_{i^\ast} > \varepsilon$.
\end{lemma}
\begin{proof}
	We first argue that the $\alpha$-envy graph is acyclic.
	Indeed, assume that there is a cycle $i_1 \to i_2 \to \dots \to i_r \to i_1$. 
	Then, a Pareto dominating assignment $Y$ can be obtained from $X$ by setting $y_{i_\ell} = x_{i_{\ell + 1}}$ for all $\ell \in \{1, \dots, r-1\}$, $y_{i_r} = x_{i_1}$, and $y_j = x_j$ for all agents not in the cycle.
	This contradicts $X \in F(\omega)$.
	
	Since $\rho \leq 1/2 < 1$, we have $\varepsilon < 1/n$.\footnote{Note that this is important for $\Omega_\varepsilon$ to be non-empty.}
	Since $\omega \in \Omega_\varepsilon$, each component $\omega_i$ is at least $\varepsilon$, and at least one of them is at least $1/n$.
    In other words, at most $n-1$ components $\omega_i$ lie in the interval $[\varepsilon, 1/n)$.
	Consider the partition of this interval into the $n+1$ sub-intervals
    \begin{equation*}
        \left[\nicefrac{\rho^{k+1}}{n}, \nicefrac{\rho^k}{n}\right), \quad k=0,1,\dots,n.
    \end{equation*}
    By the pigeonhole principle, at least one of these sub-intervals contains none of the weights $\omega_i$. 
    Denote this sub-interval by $[\gamma,\gamma/\rho)$, where $\gamma = \nicefrac{\rho^{k+1}}{n}$ is the left endpoint of the chosen sub-interval. 
	
	Further, define $D \coloneqq \{i \in N : \omega_i < \gamma\}$ and $A \coloneqq \{i \in N : \omega_i \geq \gamma/\rho\}$. 
    Since no $\omega_i$ lies in $[\gamma,\gamma/\rho)$, the sets $D$ and $A$ partition the set of agents $N$.
    On top of that, we know that 
	\begin{itemize}
		\item $A$ is non-empty since $\omega_i \geq 1/n$ for some $i \in N$, and 
		\item every $i \in A$ satisfies $\omega_i > \varepsilon$.
	\end{itemize}
	Therefore, it suffices to demonstrate the existence of a sink vertex in $A$.
	To this end, observe that by our previous argument the subgraph induced by $A$ is also acyclic, and hence contains a vertex $i^\ast$ with no outgoing edge within $A$.
	Moreover, for any two agents $i \in A$ and $j \in D$, we have $\omega_j < \gamma \leq \rho \omega_i$, so by \Cref{lemma:weight-ratio}, we have $\phi_i(x_j, x_i) \leq \alpha$.
	Hence, there cannot exist a directed edge in the $\alpha$-envy graph which runs from $A$ to $D$.
	It follows that $i^\ast$ is a sink in the entire $\alpha$-envy graph, and since $i^\ast \in A$, we also get $\omega_{i^\ast} > \varepsilon$.
\end{proof}

Using the above lemma, we can show that whenever there exists an agent who envies another agent by more than $\alpha$, the weight update map $g$ maps to a distinct vector of welfare weights, i.e., $g(X, \omega) \neq \omega$ whenever there are agents $i, j \in N$ with $\phi_i(x_j, x_i) > \alpha$.

\begin{lemma} \label{lemma:weight-decrease}
	Suppose that the random assignment $X$ is not $\alpha$-envy-free. 
	If agent $i^\ast \in N$ is a sink in the $\alpha$-envy graph of $X$ with $\omega_{i^\ast} > \varepsilon$, then $g_{i^\ast}(X, \omega) < \omega_{i^\ast}$.
\end{lemma}
\begin{proof}
	The assignment $X$ not being $\alpha$-envy-free is equivalent to $\sum_{i \in N} \nu_i(X) > 0$, and $i^\ast$ being a sink vertex is equivalent to $\nu_{i^\ast}(X) = 0$.
	For each $i \in N$, since $\mu_i(X, \omega) = \omega_i + \nu_i(X)$, $\omega_i \geq \varepsilon$ and $\nu_i(X) \geq 0$, we have $\mu_i(X, \omega) \geq \varepsilon$.
	Since $\sum_{i \in N} \omega_i = 1$, we have
	\begin{equation*}
		\sum_{i \in N} \mu_i(X, \omega) = 1 + \sum_{i \in N} \nu_i(X) > 1,
	\end{equation*}
	i.e., $\mu(X, \omega) \notin \Omega_\varepsilon$.
	It is a well-known property of the Euclidean projection onto the truncated probability simplex $\Omega_\varepsilon$ \citep[see, e.g.,][]{WaCa13a} that there exists $\tau > 0$ such that for each $i \in N$,
	\begin{equation*}
		g_i(X, \omega) = \max \{ \varepsilon, \mu_i(X, \omega) - \tau \},
	\end{equation*}
	i.e., the Euclidean projection onto $\Omega_\varepsilon$ subtracts a uniform strictly positive threshold from each component $\mu_i(X, \omega)$ under the constraints of $\Omega_\varepsilon$.
	Since $\nu_{i^\ast}(X) = 0$, we have $\mu_{i^\ast}(X, \omega) = \omega_{i^\ast}$, and hence
	\begin{equation*}
		g_{i^\ast}(X, \omega) = \max \{ \varepsilon, \omega_{i^\ast} - \tau \}.
	\end{equation*}
	Since $\omega_{i^\ast} > \varepsilon$ and $\tau > 0$, it follows that $g_{i^\ast}(X, \omega) < \omega_{i^\ast}$.
\end{proof}

Finally, we are ready to prove \Cref{proposition:fixed-point-envy-freeness}.

\begin{proof}[Proof of \Cref{proposition:fixed-point-envy-freeness}]
	Let $(X^\ast, \omega^\ast)$ be a fixed point of the correspondence $\Gamma$.
	Assume for contradiction that $X^\ast$ is not $\alpha$-envy-free, i.e., there exist agents $i, j \in N$ such that $\phi_i(x^\ast_j, x^\ast_i) > \alpha$.
	Hence, $\nu_i(X^\ast) > 0$, and $\sum_{k \in N} \nu_k(X^\ast) > 0$.
	Since $X^\ast \in F(\omega^\ast)$, by \Cref{lemma:sink-existence} there exists an agent $i^\ast$ with $\omega^\ast_{i^\ast} > \varepsilon$ which is also a sink vertex in the $\alpha$-envy graph of $X^\ast$, i.e., $\nu_{i^\ast}(X^\ast) = 0$.
	Applying \Cref{lemma:weight-decrease} gives $g_{i^\ast}(X^\ast, \omega^\ast) < \omega^\ast_{i^\ast}$, which contradicts that $(X^\ast, \omega^\ast)$ is a fixed point of $\Gamma$.
	Therefore, $X^\ast$ is $\alpha$-envy-free.
\end{proof}

We now provide a proof for \Cref{proposition:eff-not-closed}.

\SSBEfficientSetOpen*

\begin{proof}
	\begin{itemize}
	
	\item[\emph{(i)}] Let $n = 4$.
	Consider the following profile with ordinal preferences over objects in $O = \{a, b, c, d\}$, that are extended to $\Delta$ via the \pc relation (which constitutes a special case of SSB preferences; see \Cref{section:ordinal-ra}): 
		\begin{align*}
			1&: a \succ b \succ c \succ d \\
			2&: a \succ b \succ d \succ c \\
			3&: a \succ b \succ d \succ c \\
			4&: d \succ a \succ c \succ b
		\end{align*} 

	For each $\delta \in (0, \nicefrac{1}{10})$, define the weight vector $\omega^\delta \in \mathbb{R}^n_{>0}$ as 
	\begin{equation*}
		\omega^\delta = \left(\frac{1-\delta}{3}, \frac{1-\delta}{3}, \frac{1-\delta}{3}, \delta\right).
	\end{equation*}

	We now construct a convergent sequence $\left(X^\delta\right)_{\delta \in (0, \nicefrac{1}{10})}$ of random assignments such that $X^\delta \to X$ as $\delta \to 0$, and also $X^\delta \in F(\omega^\delta)$ for each $\delta \in (0, \nicefrac{1}{10})$.
	By \Cref{theorem:efficiency-welfare}, this implies that each $X^\delta$ is efficient.
	Though not required to prove the statement, we remark that our construction also illustrates \Cref{theorem:alpha-envy-free}: each $X^\delta$ is $O(\delta)$-envy-free, and hence the limit assignment $X$ is exactly envy-free.
	However, we claim that $X$ is not efficient.

	For each $\delta \in (0, \nicefrac{1}{10})$, define the random assignment 
		\begin{equation*}
			X^\delta \coloneqq
			\begin{pmatrix}
			\frac{1 - 7\delta}{3(1-\delta)} & \frac{1 + 5\delta}{3(1-\delta)} & \frac{1}{3} & 0\\
			\frac{1 + 2\delta}{3(1-\delta)} & \frac{1 - 4\delta}{3(1-\delta)} & 0 & \frac{1}{3}\\
			\frac{1 + 2\delta}{3(1-\delta)} & \frac{1 - 4\delta}{3(1-\delta)} & 0 & \frac{1}{3}\\
			0 & 0 & \frac{2}{3} & \frac{1}{3}
			\end{pmatrix}.
		\end{equation*}

	To prove each $X^\delta \in F(\omega^\delta)$, it is required to show that \[\sum_{i \in N} \omega^\delta_i \phi_i(x^\delta_i, y_i) \geq 0\] for every random assignment $Y \in \mathcal{M}$.
	We claim that this inequality holds for all $\delta \in (0, \nicefrac{1}{10})$ when $Y$ is each of the 24 deterministic assignments, omitting the explicit computations for brevity. 
	By bilinearity, this suffices to show that $X^\delta \in F(\omega^\delta)$, and hence that $X^\delta$ is efficient.
	As $\delta \to 0$, we have 
	\begin{equation*}
		X^\delta \to X = 
			\begin{pmatrix}
				\nicefrac{1}{3} & \nicefrac{1}{3} & \nicefrac{1}{3} & 0\\
				\nicefrac{1}{3} & \nicefrac{1}{3} & 0 & \nicefrac{1}{3}\\
				\nicefrac{1}{3} & \nicefrac{1}{3} & 0 & \nicefrac{1}{3}\\
				0 & 0 & \nicefrac{2}{3} & \nicefrac{1}{3}
			\end{pmatrix}.
	\end{equation*}
	To see that $X$ is not efficient, consider the random assignment
	\begin{equation*}
		Y = 
			\begin{pmatrix}
				\nicefrac{1}{2} & 0 & \nicefrac{1}{2} & 0\\
				0 & 1 & 0 & 0\\
				\nicefrac{1}{2} & 0 & 0 & \nicefrac{1}{2}\\
				0 & 0 & \nicefrac{1}{2} & \nicefrac{1}{2}
			\end{pmatrix}.
	\end{equation*}
	Then,
	\begin{equation*}
			\phi_1(y_1,x_1)=0,
			\qquad
			\phi_2(y_2,x_2)=0,
			\qquad
			\phi_3(y_3,x_3)=0,
			\qquad
			\phi_4(y_4,x_4)=\nicefrac{1}{6}.
		\end{equation*}
	Hence, $Y$ Pareto dominates $X$.

	We have exhibited a convergent sequence $\left(X^\delta\right)$ of efficient random assignments which converges to a Pareto dominated random assignment $X$.
	Hence, the set of efficient random assignments is not closed. 

	We also remark, as mentioned above, that each $X^\delta$ is $O(\delta)$-envy-free, and hence the limit assignment $X^\ast$ is exactly envy-free.
	We omit the proof of this fact.

	\item[\emph{(ii)}] Let $(X^k)$ be a convergent sequence of weakly efficient random assignments with $X^k \to X$ as $k \to \infty$.
    Assume for contradiction that $X$ is not weakly efficient, and let $Y$ be a random assignment such that $\phi_i(x_i, y_i) < 0$ for all agents $i \in N$.
    Since the map $X' \mapsto \phi_i(x'_i, y_i)$ is continuous for each $i \in N$, there exists a sufficiently large integer $k$ such that $\phi_i(x^k_i, y_i) < 0$ for each $i \in N$.
    This contradicts the weak efficiency of $X^k$.
	\end{itemize}
\end{proof}


\begin{thebibliography}{78}
\providecommand{\natexlab}[1]{#1}
\providecommand{\url}[1]{\texttt{#1}}
\expandafter\ifx\csname urlstyle\endcsname\relax
  \providecommand{\doi}[1]{doi: #1}\else
  \providecommand{\doi}{doi: \begingroup \urlstyle{rm}\Url}\fi

\bibitem[Abdulkadiro{\u g}lu and S{\"o}nmez(1998)]{AbSo98a}
A.~Abdulkadiro{\u g}lu and T.~S{\"o}nmez.
\newblock Random serial dictatorship and the core from random endowments in
  house allocation problems.
\newblock \emph{Econometrica}, 66\penalty0 (3):\penalty0 689--701, 1998.

\bibitem[Allais(1953)]{Alla53a}
M.~Allais.
\newblock Le comportement de l'homme rationnel devant le risque: Critique des
  postulats et axiomes de l'ecole americaine.
\newblock \emph{Econometrica}, 21\penalty0 (4):\penalty0 503--546, 1953.

\bibitem[Anand(1993)]{Anan93a}
P.~Anand.
\newblock The philosophy of intransitive preference.
\newblock \emph{The Economic Journal}, 103\penalty0 (417):\penalty0 337--346,
  1993.

\bibitem[Anand(2009)]{Anan09a}
P.~Anand.
\newblock Rationality and intransitive preference: {F}oundations for the modern
  view.
\newblock In P.~Anand, P.~K. Pattanaik, and C.~Puppe, editors, \emph{The
  Handbook of Rational and Social Choice}, chapter~6. Oxford University Press,
  2009.

\bibitem[Arrow(1951)]{Arro51a}
K.~J. Arrow.
\newblock \emph{Social Choice and Individual Values}.
\newblock New Haven: Cowles Foundation, 1st edition, 1951.
\newblock 2nd edition 1963.

\bibitem[Arrow and Debreu(1954)]{ArDe54a}
K.~J. Arrow and G.~Debreu.
\newblock Existence of an equilibrium for a competitive economy.
\newblock \emph{Econometrica}, 22\penalty0 (3):\penalty0 265--290, 1954.

\bibitem[Azevedo and Budish(2019)]{AzBu19a}
E.~M. Azevedo and E.~Budish.
\newblock Strategyproofness in the large.
\newblock \emph{Review of Economic Studies}, 86\penalty0 (1):\penalty0 81--116,
  2019.

\bibitem[Aziz et~al.(2013)Aziz, Brandt, and Stursberg]{ABS13a}
H.~Aziz, F.~Brandt, and P.~Stursberg.
\newblock On popular random assignments.
\newblock In \emph{Proceedings of the 6th International Symposium on
  Algorithmic Game Theory (SAGT)}, volume 8146 of \emph{Lecture Notes in
  Computer Science (LNCS)}, pages 183--194. Springer-Verlag, 2013.

\bibitem[Aziz et~al.(2015)Aziz, Brandl, and Brandt]{ABB14b}
H.~Aziz, F.~Brandl, and F.~Brandt.
\newblock Universal {P}areto dominance and welfare for plausible utility
  functions.
\newblock \emph{Journal of Mathematical Economics}, 60:\penalty0 123--133,
  2015.

\bibitem[Aziz et~al.(2018)Aziz, Brandl, Brandt, and Brill]{ABBB15a}
H.~Aziz, F.~Brandl, F.~Brandt, and M.~Brill.
\newblock On the tradeoff between efficiency and strategyproofness.
\newblock \emph{Games and Economic Behavior}, 110:\penalty0 1--18, 2018.

\bibitem[{Bar-Hillel} and Margalit(1988)]{BaMa88a}
M.~{Bar-Hillel} and A.~Margalit.
\newblock How vicious are cycles of intransitive choice?
\newblock \emph{Theory and Decision}, 24\penalty0 (2):\penalty0 119--145, 1988.

\bibitem[Berge(1963)]{Berg63a}
C.~Berge.
\newblock \emph{Topological Spaces: Including a Treatment of Multi-Valued
  Functions, Vector Spaces and Convexity}.
\newblock Oliver and Boyd, 1963.

\bibitem[Bergstrom(1992)]{Berg92a}
T.~C. Bergstrom.
\newblock When non-transitive relations take maxima and competitive equilibrium
  can't be beat.
\newblock In W.~Neuefeind and R.~G. Riezmann, editors, \emph{Economic Theory
  and International Trade (Essays in Memoriam of J. Trout Rader)}, pages
  29--52. Springer-Verlag, 1992.

\bibitem[Birkhoff(1946)]{Birk46a}
G.~Birkhoff.
\newblock Three observations on linear algebra.
\newblock \emph{Univ. Nac. Tacuman Rev. Ser. A}, 5:\penalty0 147--151, 1946.

\bibitem[Blavatskyy(2006)]{Blav06a}
P.~R. Blavatskyy.
\newblock Axiomatization of a preference for most probable winner.
\newblock \emph{Theory and Decision}, 60\penalty0 (1):\penalty0 17--33, 2006.

\bibitem[Blyth(1972)]{Blyt72a}
C.~R. Blyth.
\newblock Some probability paradoxes in choice from among random alternatives.
\newblock \emph{Journal of the American Statistical Association}, 67\penalty0
  (338):\penalty0 366--373, 1972.

\bibitem[Bogomolnaia and Moulin(2001)]{BoMo01a}
A.~Bogomolnaia and H.~Moulin.
\newblock A new solution to the random assignment problem.
\newblock \emph{Journal of Economic Theory}, 100\penalty0 (2):\penalty0
  295--328, 2001.

\bibitem[Brandl(2013)]{Bran13b}
F.~Brandl.
\newblock Efficiency and incentives in randomized social choice.
\newblock Master's thesis, Technische Universit{\"a}t M{\"u}nchen, 2013.

\bibitem[Brandl and Brandt(2020)]{BrBr17a}
F.~Brandl and F.~Brandt.
\newblock Arrovian aggregation of convex preferences.
\newblock \emph{Econometrica}, 88\penalty0 (2):\penalty0 799--844, 2020.

\bibitem[Brandl et~al.(2016)Brandl, Brandt, and Seedig]{Bran13a}
F.~Brandl, F.~Brandt, and H.~G. Seedig.
\newblock Consistent probabilistic social choice.
\newblock \emph{Econometrica}, 84\penalty0 (5):\penalty0 1839--1880, 2016.

\bibitem[Brandl et~al.(2018)Brandl, Brandt, Eberl, and Geist]{BBEG16a}
F.~Brandl, F.~Brandt, M.~Eberl, and C.~Geist.
\newblock Proving the incompatibility of efficiency and strategyproofness via
  {SMT} solving.
\newblock \emph{Journal of the ACM}, 65\penalty0 (2):\penalty0 1--28, 2018.

\bibitem[Brandl et~al.(2019)Brandl, Brandt, and Hofbauer]{BBH15c}
F.~Brandl, F.~Brandt, and J.~Hofbauer.
\newblock Welfare maximization entices participation.
\newblock \emph{Games and Economic Behavior}, 14:\penalty0 308--314, 2019.

\bibitem[Brandt et~al.(2017)Brandt, Hofbauer, and Suderland]{BHS17a}
F.~Brandt, J.~Hofbauer, and M.~Suderland.
\newblock Majority graphs of assignment problems and properties of popular
  random assignments.
\newblock In \emph{Proceedings of the 16th International Conference on
  Autonomous Agents and Multiagent Systems (AAMAS)}, pages 335--343, 2017.

\bibitem[Brandt et~al.(2023)Brandt, Lederer, and Suksompong]{BLS22c}
F.~Brandt, P.~Lederer, and W.~Suksompong.
\newblock Incentives in social decision schemes with pairwise comparison
  preferences.
\newblock \emph{Games and Economic Behavior}, 142:\penalty0 266--291, 2023.

\bibitem[Budish(2011)]{Budi11a}
E.~Budish.
\newblock The combinatorial assignment problem: Approximate competitive
  equilibrium from equal incomes.
\newblock \emph{Journal of Political Economy}, 119\penalty0 (6):\penalty0
  1061--1103, 2011.

\bibitem[Butler and Pogrebna(2018)]{BuPo18a}
D.~Butler and G.~Pogrebna.
\newblock Predictably intransitive preferences.
\newblock \emph{Judgment and Decision Making}, 13\penalty0 (3):\penalty0
  217--236, 2018.

\bibitem[Carroll(2010)]{Carr10a}
G.~D. Carroll.
\newblock An efficiency theorem for incompletely known preferences.
\newblock \emph{Journal of Economic Theory}, 145\penalty0 (6):\penalty0
  2463--2470, 2010.

\bibitem[Chen et~al.(2022)Chen, Chen, Peng, and Yannakakis]{CCPY22a}
T.~Chen, X.~Chen, B.~Peng, and M.~Yannakakis.
\newblock Computational hardness of the hylland-zeckhauser scheme.
\newblock In \emph{Proceedings of the 33rd Annual ACM-SIAM Symposium on
  Discrete Algorithms (SODA)}, pages 2253--2268, 2022.

\bibitem[Cho and Dogan(2016)]{ChDo16a}
W.~J. Cho and B.~Dogan.
\newblock Equivalence of efficiency notions for ordinal assignment problems.
\newblock \emph{Economics Letters}, 146:\penalty0 8--12, 2016.

\bibitem[Cole and Tao(2021)]{CoTa21a}
R.~Cole and Y.~Tao.
\newblock On the existence of pareto efficient and envy-free allocations.
\newblock \emph{Journal of Economic Theory}, 193, 2021.

\bibitem[Debreu(1959)]{Debr59a}
G.~Debreu.
\newblock \emph{Theory of Value. An Axiomatic Analysis of Economic
  Equilibrium}, volume~17 of \emph{Cowles Foundation for Research in Economics
  at Yale University}.
\newblock Wiley and Sons, 1959.

\bibitem[Echenique et~al.(2021)Echenique, Miralles, and Zhang]{EMZ21a}
F.~Echenique, A.~Miralles, and J.~Zhang.
\newblock Constrained pseudo-market equilibrium.
\newblock \emph{American Economic Review}, 111\penalty0 (11):\penalty0
  3699--3732, 2021.

\bibitem[Fishburn(1970)]{Fish70c}
P.~C. Fishburn.
\newblock The irrationality of transitity in social choice.
\newblock \emph{Behavioral Science}, 15:\penalty0 119--123, 1970.

\bibitem[Fishburn(1982)]{Fish82c}
P.~C. Fishburn.
\newblock Nontransitive measurable utility.
\newblock \emph{Journal of Mathematical Psychology}, 26\penalty0 (1):\penalty0
  31--67, 1982.

\bibitem[Fishburn(1984{\natexlab{a}})]{Fish84a}
P.~C. Fishburn.
\newblock Probabilistic social choice based on simple voting comparisons.
\newblock \emph{Review of Economic Studies}, 51\penalty0 (4):\penalty0
  683--692, 1984{\natexlab{a}}.

\bibitem[Fishburn(1984{\natexlab{b}})]{Fish84c}
P.~C. Fishburn.
\newblock {SSB} utility theory: {A}n economic perspective.
\newblock \emph{Mathematical Social Sciences}, 8\penalty0 (1):\penalty0 63--94,
  1984{\natexlab{b}}.

\bibitem[Fishburn(1984{\natexlab{c}})]{Fish84d}
P.~C. Fishburn.
\newblock Dominance in {SSB} utility theory.
\newblock \emph{Journal of Economic Theory}, 34\penalty0 (1):\penalty0
  130--148, 1984{\natexlab{c}}.

\bibitem[Fishburn(1988)]{Fish88a}
P.~C. Fishburn.
\newblock \emph{Nonlinear preference and utility theory}.
\newblock The Johns Hopkins University Press, 1988.

\bibitem[Fishburn(1991)]{Fish91a}
P.~C. Fishburn.
\newblock Nontransitive preferences in decision theory.
\newblock \emph{Journal of Risk and Uncertainty}, 4\penalty0 (2):\penalty0
  113--134, 1991.

\bibitem[Fishburn and Rosenthal(1986)]{FiRo86a}
P.~C. Fishburn and R.~W. Rosenthal.
\newblock Noncooperative games and nontransitive preferences.
\newblock \emph{Mathematical Social Sciences}, 12\penalty0 (1):\penalty0 1--7,
  1986.

\bibitem[Gale and {Mas-Colell}(1975)]{GaMA75a}
D.~Gale and A.~{Mas-Colell}.
\newblock An equilibrium existence theorem for a general model without ordered
  preferences.
\newblock \emph{Journal of Mathematical Economics}, 2:\penalty0 9--15, 1975.

\bibitem[Garg et~al.(2026)Garg, Tao, and V{\'e}gh]{GTV26a}
J.~Garg, Y.~Tao, and L.~A. V{\'e}gh.
\newblock Tight efficiency bounds for the probabilistic serial and related
  mechanisms.
\newblock 2026.
\newblock Working paper.

\bibitem[Gibbard(1977)]{Gibb77a}
A.~Gibbard.
\newblock Manipulation of schemes that mix voting with chance.
\newblock \emph{Econometrica}, 45\penalty0 (3):\penalty0 665--681, 1977.

\bibitem[Hara et~al.(2019)Hara, Ok, and Riella]{HOR19a}
K.~Hara, E.~A. Ok, and G.~Riella.
\newblock Coalitional expected multi-utility theory.
\newblock \emph{Econometrica}, 87\penalty0 (3):\penalty0 933--980, 2019.

\bibitem[He et~al.(2018)He, Miralles, Pycia, and Yan]{HMPY18a}
Y.~He, A.~Miralles, M.~Pycia, and J.~Yan.
\newblock A pseudo-market approach to allocation with priorities.
\newblock \emph{American Economic Journal: Microeconomics}, 10\penalty0
  (3):\penalty0 272--314, 2018.

\bibitem[Hylland(1980)]{Hyll80a}
A.~Hylland.
\newblock Strategyproofness of voting procedures with lotteries as outcomes and
  infinite sets of strategies.
\newblock {M}imeo, 1980.

\bibitem[Hylland and Zeckhauser(1979)]{HyZe79a}
A.~Hylland and R.~Zeckhauser.
\newblock The efficient allocation of individuals to positions.
\newblock \emph{The Journal of Political Economy}, 87\penalty0 (2):\penalty0
  293--314, 1979.

\bibitem[Kahneman and Tversky(1979)]{KaTv79a}
D.~Kahneman and A.~Tversky.
\newblock Prospect theory: {A}n analysis of decision under risk.
\newblock \emph{Econometrica}, 47\penalty0 (2):\penalty0 263--292, 1979.

\bibitem[Katta and Sethuraman(2006)]{KaSe06a}
A.-K. Katta and J.~Sethuraman.
\newblock A solution to the random assignment problem on the full preference
  domain.
\newblock \emph{Journal of Economic Theory}, 131\penalty0 (1):\penalty0
  231--250, 2006.

\bibitem[Kavitha et~al.(2011)Kavitha, Mestre, and Nasre]{KMN11a}
T.~Kavitha, J.~Mestre, and M.~Nasre.
\newblock Popular mixed matchings.
\newblock \emph{Theoretical Computer Science}, 412\penalty0 (24):\penalty0
  2679--2690, 2011.

\bibitem[Kuhn(1955)]{Kuhn55a}
H.~W. Kuhn.
\newblock The {H}ungarian method for the assignment problem.
\newblock \emph{Naval Research Logistics Quarterly}, 2\penalty0
  (1--2):\penalty0 83--97, 1955.

\bibitem[Llinares(1998)]{Llin98a}
J.-V. Llinares.
\newblock Unified treatment of the problem of existence of maximal elements in
  binary relations: {A} characterization.
\newblock \emph{Journal of Mathematical Economics}, 29\penalty0 (3):\penalty0
  285--302, 1998.

\bibitem[Machina(1983)]{Mach83a}
M.~J. Machina.
\newblock Generalized expected utility analysis and the nature of observed
  violations of the independence axiom.
\newblock In B.~Stigum and F.~Wenstop, editors, \emph{Foundations of Utility
  and Risk Theory with Applications}, chapter~5. Springer, 1983.

\bibitem[Machina(1989)]{Mach89a}
M.~J. Machina.
\newblock Dynamic consistency and non-expected utility models of choice under
  uncertainty.
\newblock \emph{Journal of Economic Literature}, 27\penalty0 (4):\penalty0
  1622--1668, 1989.

\bibitem[{Mas-Colell}(1974)]{MasC74a}
A.~{Mas-Colell}.
\newblock An equilibrium existence theorem without complete or transitive
  preferences.
\newblock \emph{Journal of Mathematical Economics}, 1:\penalty0 237--247, 1974.

\bibitem[Mas-Colell(1992)]{MasC92a}
A.~Mas-Colell.
\newblock Equilibrium theory with possibly satiated preferences.
\newblock In M.~Majumdar, editor, \emph{Equilibrium and Dynamics: Essays in
  Honour of David Gale}, pages 201--213. Palgrave Macmillan UK, 1992.

\bibitem[May(1954)]{May54a}
K.~May.
\newblock Intransitivity, utility, and the aggregation of preference patters.
\newblock \emph{Econometrica}, 22\penalty0 (1):\penalty0 1--13, 1954.

\bibitem[McClennen(1988)]{McCl88a}
E.~F. McClennen.
\newblock Sure-thing doubts.
\newblock In P.~G{\"a}rdenfors and N.-E. Sahlin, editors, \emph{Decision,
  Probability and Utility}, chapter~10. Cambridge University Press, 1988.

\bibitem[Miralles(2017)]{Mira17a}
A.~Miralles.
\newblock Ex-ante efficiency in assignments with seniority rights.
\newblock \emph{Review of Economic Design}, 21\penalty0 (1):\penalty0 33--48,
  2017.

\bibitem[Morawski(2022)]{Mora22a}
P.~Morawski.
\newblock Random assignment with pairwise comparison preferences.
\newblock Bachelor's thesis, Technische Universit{\"a}t M{\"u}nchen, 2022.

\bibitem[Moulin(1988)]{Moul88b}
H.~Moulin.
\newblock Condorcet's principle implies the no show paradox.
\newblock \emph{Journal of Economic Theory}, 45\penalty0 (1):\penalty0 53--64,
  1988.

\bibitem[Nash(1950)]{Nash50a}
J.~F. Nash.
\newblock Equilibrium points in $n$-person games.
\newblock \emph{Proceedings of the National Academy of Sciences (PNAS)},
  36:\penalty0 48--49, 1950.

\bibitem[Packard(1982)]{Pack82a}
D.~J. Packard.
\newblock Cyclical preference logic.
\newblock \emph{Theory and Decision}, 14\penalty0 (4):\penalty0 415--426, 1982.

\bibitem[Postlewaite and Schmeidler(1986)]{PoSc86a}
A.~Postlewaite and D.~Schmeidler.
\newblock Strategic behaviour and a notion of ex ante efficiency in a voting
  model.
\newblock \emph{Social Choice and Welfare}, 3\penalty0 (1):\penalty0 37--49,
  1986.

\bibitem[Rubinstein and Segal(2012)]{RuSe12a}
A.~Rubinstein and U.~Segal.
\newblock On the likelihood of cyclic comparisons.
\newblock \emph{Journal of Economic Theory}, 147\penalty0 (6):\penalty0
  2483--2491, 2012.

\bibitem[Sato(2010)]{Sato10b}
N.~Sato.
\newblock Satiation and existence of competitive equilibrium.
\newblock \emph{Journal of Mathematical Economics}, 46\penalty0 (4):\penalty0
  534--551, 2010.

\bibitem[Shafer and Sonnenschein(1975)]{ShSo75a}
W.~J. Shafer and H.~Sonnenschein.
\newblock Equilibrium in abstract economies without ordered preferences.
\newblock \emph{Journal of Mathematical Economics}, 2\penalty0 (3):\penalty0
  345--348, 1975.

\bibitem[Sonnenschein(1971)]{Sonn71a}
H.~Sonnenschein.
\newblock Demand theory without transitive preference with applications to the
  theory of competitive equilibrium.
\newblock In J.~Chipman, L.~Hurwicz, M.~Richter, and H.~Sonnenschein, editors,
  \emph{Preferences, Utility and Demand}. Houghton Mifflin Harcourt, 1971.

\bibitem[Steinhaus and Trybula(1959)]{StTr59a}
H.~Steinhaus and S.~Trybula.
\newblock On a paradox in applied probabilities.
\newblock \emph{Bulletin of the Polish Academy of Sciences}, 7:\penalty0
  67--69, 1959.

\bibitem[Tr{\"o}bst and Vazirani(2026)]{TrVa26a}
T.~Tr{\"o}bst and V.~V. Vazirani.
\newblock Cardinal-utility matching markets: {T}he quest for envy-freeness,
  pareto-optimality, and efficient computability.
\newblock \emph{Mathematics of Operations Research}, 2026.
\newblock Forthcoming.

\bibitem[Vazirani and Yannakakis(2025)]{VaYa25a}
V.~V. Vazirani and M.~Yannakakis.
\newblock Computational complexity of the hylland--zeckhauser mechanism for
  one-sided matching markets.
\newblock \emph{SIAM Journal on Computing}, 54\penalty0 (2):\penalty0 193--232,
  2025.

\bibitem[{von Neumann}(1928)]{vNeu28a}
J.~{von Neumann}.
\newblock Zur {T}heorie der {G}esellschaftspiele.
\newblock \emph{Mathematische Annalen}, 100\penalty0 (1):\penalty0 295--320,
  1928.

\bibitem[{von Neumann}(1953)]{vNeu53a}
J.~{von Neumann}.
\newblock A certain zero-sum two-person game equivalent to the optimal
  assignment problem.
\newblock In \emph{Contributions to the Theory of Games II}, number~28 in
  Annals of Mathematics Studies, pages 5--12. Princeton University Press, 1953.

\bibitem[von Neumann and Morgenstern(1947)]{vNM47a}
J.~von Neumann and O.~Morgenstern.
\newblock \emph{Theory of Games and Economic Behavior}.
\newblock Princeton University Press, 2nd edition, 1947.

\bibitem[Walker(1979)]{Walk79a}
M.~Walker.
\newblock A generalization of the maximum theorem.
\newblock \emph{International Economic Review}, 20\penalty0 (1):\penalty0
  267--272, 1979.

\bibitem[Wang and Carreira-Perpi{\~n}{\'a}n(2013)]{WaCa13a}
W.~Wang and M.~{\'A}. Carreira-Perpi{\~n}{\'a}n.
\newblock Projection onto the probability simplex: An efficient algorithm with
  a simple proof, and an application.
\newblock Technical report, https://arxiv.org/pdf/1309.1541, 2013.

\bibitem[Yan and Liu(2026)]{YaLi26a}
Y.~Yan and Z.~Liu.
\newblock Constant approximation for {H}ylland--{Z}eckhauser equilibria.
\newblock Technical report, https://arxiv.org/abs/2606.06317, 2026.

\bibitem[Zhou(1990)]{Zhou90a}
L.~Zhou.
\newblock On a conjecture by {G}ale about one-sided matching problems.
\newblock \emph{Journal of Economic Theory}, 52\penalty0 (1):\penalty0
  123--135, 1990.

\end{thebibliography}
\end{document}